\documentclass[onecolumn,10pt]{IEEEtran}
\usepackage{graphicx}
\usepackage{amsmath,amssymb}
\usepackage{cases}
\usepackage{color}
\usepackage{amsfonts}
\usepackage{indentfirst}
\usepackage{slashbox}
\usepackage{float}
\usepackage{cite}
\usepackage{caption}
\usepackage{algorithm}
\usepackage{algpseudocode}
\usepackage{booktabs}
\usepackage{subfigure}
\usepackage{multirow}
\usepackage{amsthm}
\makeatletter
\newtheoremstyle{mythm}{3pt}{3pt}{}{16pt}{\bfseries}{:}{.5em}{}
\theoremstyle{mythm}
\newtheorem{theorem}{Theorem}
\newtheorem{example}{Example}
\newtheorem{definition}{Definition}

\newtheorem{proposition}{Proposition}

\newtheorem{lemma}{Lemma}
\newtheorem{construction}{Construction}

\newtheorem{partition}{Partition}
\newcommand{\tabincell}[2]{\begin{tabular}{@{}#1@{}}#2\end{tabular}}
\begin{document}
\title{Coded Caching Schemes with Linear Subpacketizations
\author{Xi Zhong, Minquan Cheng, Ruizhong Wei}
\thanks{Zhong and Cheng are with Guangxi Key Lab of Multi-source Information Mining $\&$ Security, Guangxi Normal University,
Guilin 541004, China, (e-mail: Zhong19961225@outlook.com, chengqinshi@hotmail.com).}
\thanks{R. Wei is with Department of Computer Science, Lakehead University, Thunder Bay, ON, Canada, P7B 5E1,(e-mail: rwei@lakeheadu.ca).}
}

\date{}
\maketitle

\begin{abstract}
In coded caching system we prefer to design a coded caching scheme with low subpacketization and small transmission rate (i.e., the low implementation complexity and the efficient transmission during the peak traffic times). Placement delivery arrays (PDA) can be used to design code caching schemes. In this paper we propose a framework of constructing PDAs via Hamming distance. As an application, two classes of coded caching schemes with linear subpacketizations and small transmission rates are obtained.
\end{abstract}


%
\IEEEpeerreviewmaketitle
\section{Introduction}
\label{introduction}
\IEEEPARstart {T}{he} immense growth in wireless data traffic leads to an enormous pressure on wireless network especially the high temporal variability of network traffic results in congestion during the peak traffic times while underutilization during the off-peak times.
In order to make fully use of the local caching to solve this problem, coded caching system was proposed in \cite{MN} which can significantly   reduce the amount of transmission during the peak traffic time.

In an $F$-division $(K,M,N)$ centralized coded caching system, a server containing $N$ files with equal size connects to $K$ users, each of which has memory of size $M$ files, through an error-free shared link. This system consists of two phases, i.e., the placement phase during the off-peak traffic times and the delivery phase during the peak traffic times. In the placement phase, without knowledge of later demands, the server divides each file into $F$ packets with equal size where $F$ is referred as subpacketization, and then places some contents generated by the packets of all the files to each user's cache memory. In the delivery phase, each user requires one file randomly. Then the server sends some coded signals with the size of at most $R$ files ($R$ is referred as transmission), which satisfies various demands from users.

The first coded caching scheme proposed by Maddah-Ali and Niesen in \cite{MN}, which is referred as MN scheme in this paper, achieves the approximatively minimum transmission rate when $N>K$. The MN scheme has been extended to numerous models, such as Device-to-Device (D2D) caching systems \cite{JCM}, online caching \cite{PMN}, hierarchical caching \cite{KNAD}, secure caching \cite{STC1} and so on. While the subpacketization of the MN scheme increases exponentially with the number of users $K$, which leads to high implementing complexity and infeasibility in reality.

\subsection{Previously known results}
It is well known that there exists a tradeoff between subpacketization and transmission rate. Hence reducing the subpacketization must be at the cost of increasing the transmission rate compared with MN scheme.

The authors in \cite{YCTC} proposed an $F\times K$ array called placement delivery array (PDA) to generate an $F$-division coded caching scheme with $K$ users. By constructing PDAs, they obtained two classes of schemes with lower subpacketizations compared with the MN scheme. Apart from PDA, there are many other characterizations of coded caching schemes such as hypergraphs \cite{SZG}, strong edge coloring of bipartite graphs \cite{YTCC}, Ruzsa-Szem\'{e}redi graphs \cite{STD}, combinatorial design theory \cite{CD,TR}, line graphs \cite{K} and so on.  We list most of the previously known deterministic schemes, which have advantages on the  subpacketization or the transmission rate, in Table \ref{tab-known-1}.
{\begin{table*}[!htbp]
\center
\caption{Summary of some known coded caching schemes where all the variables are positive integers unless otherwise stated.\label{tab-known-1}}
\small{
\begin{tabular}{|c|c|c|c|c|c|}
\hline
References & Number of Users $K$  & Cache Fraction $\frac{M}{N}$
& Rate $R$   & Subpacketization $F$    \\ \hline

\cite{MN}: $\frac{K M}{N} \in Z^{+}$& $K$&$\frac{M}{N}$& $\frac{K(1-\frac{M}{N})}{K\frac{M}{N}+1}$&${K \choose K\frac{M}{N}}$\\ \hline

\tabincell{c}{\cite{YTCC}: $a,b<m$\\ $\lambda < \min{\{a,b\}}$}& ${m \choose a}$& $\frac{{a \choose \lambda}{m-a \choose b-\lambda}}{{m \choose a}}$& $\frac{{m \choose a+b-2\lambda}{a+b-2\lambda \choose a-\lambda}}{{m \choose b}}$
&${m \choose b}$\\ \hline

\tabincell{c}{\cite{CK}: $n+m+1$\\ $\leq z=k-t$, \\ prime power $q$}
& $\frac{q^{\frac{n(n+1)}{2}}\prod\limits_{i=0}^{n}\left[z+1-i \atop 1\right]_q}{(n+1)!}$
&\tabincell{c}{ $1-q^{(m+1)(n+1)}\cdot$\\ $\prod\limits_{i=0}^{n}\frac{\left[z-m-i \atop 1\right]_q}{\left[z+1-i \atop 1\right]_q}$}
&\tabincell{c}{ $\frac{(m+1)!q^{(n+1)(\frac{n}{2}+m+1)}}{(m+n+2)!}\cdot$\\ $\prod\limits_{i=m+1}^{m+n+1}\left[z+1-i \atop 1\right]_q$}
& $\frac{q^{\frac{m(m+1)}{2}}\prod\limits_{i=0}^{m}\left[z+1-i \atop 1\right]_q}{(m+1)!}$
\\ \hline

\cite{CD}: $(v,k,2)$-SBIBD&$v$ &$1-\frac{k-1}{v}$ &$1$ &$kv$\\ \hline

\cite{CWZW}: $t\le m$ &${m\choose t}{q}^{t}$ &$1-(\frac{q-1}{q})^{t}$&$(q-1)^{t}$
&$q^{m-1}$\\ \hline

\tabincell{c}{\cite{CWZW}: $t\le m$, $[m,m-t]_q$ \\maxium distance \\separable code}&${m\choose t}{q}^{t}$ &$1-(\frac{q-1}{q})^{t}$&$q^t-1$
&$q^{m-t}$\\ \hline
%

\cite{SJTLD}: $\frac{K M}{N}, \frac{K}{g\lceil\frac{N}{M}\rceil}\in Z^{+}$& $K$&$\frac{M}{N}$& $\frac{K}{g+1}(1-\frac{1}{\lceil\frac{M}{N}\rceil})$&$\mathcal{O}(e^g)$\\ \hline

\cite{CJTY}: $t<k$
& ${k \choose t+1}$
& $1-\frac{t+1}{{k \choose t}}$
& $\frac{k}{{k \choose t}}$
&${k \choose t}$ \\ \hline

\end{tabular}}
\end{table*} }

In \cite{STD}, it was pointed out that all the deterministic coded caching schemes introduced above can be represented by PDAs. Hence constructing appropriate PDAs makes great sense to coded caching. There are some known constructions of PDAs from view points of combinatorial designs \cite{CJTY,CJYT,CJWY,CWZW}, bipartite graphs \cite{MW} and concatenating construction \cite{SCS,ZCJ} so on. It is worth noting that the framework of constructing coded caching schemes proposed in \cite{CWZW} can include most of the previously known schemes. Furthermore based on some special PDAs, the authors in \cite{ZCWZ} generated some improved schemes with smaller subpacketizations and memory sizes compared with the scheme generated by the method in \cite{YCTC}.

\subsection{Contributions and arrangement of this paper}
In this paper we focus on linear subpacketization schemes with small transmission rates when $N>K$, where linear means linear to the number of users. Firstly we propose a framework of constructing PDAs using the concept of Hamming distance.
Secondly, we obtain two classes of coded caching schemes with
linear subpacketizations. Our new schemes have advantages on number of users, subpacketizations, memory size and transmission rates compared with several previously known deterministic schemes with linear subpacketizations.

The rest of the paper is organized as follows. In Section \ref{sec_prob} we state the preliminaries about coded caching, placement delivery array and their relationship. In Section \ref{se-characterization} we introduce the framework of constructing PDAs via Hamming distance. Two classes of schemes are obtained in Sections \ref{con-q=2} and \ref{con-q=3} respectively. The performance analysis of our new schemes is proposed in Section \ref{comparison}. Finally conclusion is drawn in Section \ref{conclusion}.

\section{Preliminaries}\label{sec_prob}
In this paper, we use bold capital letter, bold lower case letter and curlicue letter to denote array, vector and set respectively. For any positive integers $m$ and $t$ with $t< m$, let $[0,m)=\{0,1,\ldots,m-1\}$.

\subsection{Centralized coded caching system}
In a centralized coded caching system, a sever containing $N$ files, denoted by $\mathcal{W}=\{W_n\ |\ n\in[0,N)\}$, links to $K$ users, denoted by $\mathcal{K}=[0,K)$ with $K<N$ through an error-free shared link. Assume that each user has a memory of size $M$ files with $M<N$.
An $F$-division $(K,M,N)$ coded caching scheme operates in two phases which can be sketched as follows:
\begin{enumerate}
\item \textbf{Placement Phase:} All the files are divided into $F$ equal packets where $F$ is referred as subpacketization\footnote{Memory sharing technique may lead to non equally divided packets \cite{MN}, in this paper, we will not discuss this case.}
      , i.e., $\mathcal{W}=\{W_{n,j}\ |\ j\in [0,F), n\in [0,N)\}$. Each user caches some coded packets or uncoded packets from $\mathcal{W}$. $\mathcal{Z}_k$ denotes the content cached by user $k$. The size of $\mathcal{Z}_k$ is the capacity of each user's cache memory size $M$.
\item \textbf{Delivery Phase:} Each user requests one file from $\mathcal{W}$ randomly. Denote the requested file numbers by $\mathbf{d}=(d_0,d_1,\cdots,d_{K-1})$, i.e., user $k$ requests file $W_{d_k}$, where $k\in \mathcal{K}, d_k\in[0,N)$. The server broadcasts coded signals of size at most $R$ files to users, so that each user is able to recover his requested file with help of its caching contents. $R$ is called the transmission rate.
\end{enumerate}

Clearly the efficiency of the data transmission in the delivery phase increases with decreasing transmission rate. Furthermore, the complexity of the implementing a coded caching scheme increases as subpacketization $F$. So we prefer to design a scheme with both subpacketization and transmission rate as small as possible.
\subsection{Placement delivery array and two realization strategies of coded caching scheme}
\label{PDA-CCS}
Yan et al. in \cite{YCTC} first proposed the concept of placement delivery array and a realization strategy of characterizing the placement phase and delivery phase simultaneously.
\begin{definition}\rm(\cite{YCTC})
\label{def-PDA}
For positive integers $K$ and $F$, an $F\times K$ array $\mathbf{P}=(p_{i,j})$, $i\in [0,F), j\in[0,K)$, composed of a specific symbol $``*"$ called star and $S$ symbols $\{0,1,\ldots,S-1\}$, is called a $(K,F,S)$ placement delivery array (PDA) if it satisfies C$1$ in the following conditions:
\begin{enumerate}
  \item [C$1$.] For any two distinct entries $p_{i_1,j_1}$ and $p_{i_2,j_2}$,    $p_{i_1,j_1}=p_{i_2,j_2}=s\in\{0,1,\ldots,S-1\}$ only if
  \begin{enumerate}
     \item [a.] $i_1\ne i_2$, $j_1\ne j_2$, i.e., they lie in distinct rows and distinct columns;
     \item [b.] $p_{i_1,j_2}=p_{i_2,j_1}=*$, i.e., the corresponding $2\times 2$  subarray formed by rows $i_1,i_2$ and columns $j_1,j_2$ must be one of the following forms
  \begin{eqnarray*}
    \left(\begin{array}{cc}
      s & *\\
      * & s
    \end{array}\right)~\textrm{or}~
    \left(\begin{array}{cc}
      * & s\\
      s & *
    \end{array}\right).
  \end{eqnarray*}
   \end{enumerate}
  \end{enumerate}
\end{definition}
For any positive integer $Z\leq F$, $\mathbf{P}$ is denoted by $(K,F,Z,S)$ PDA if
\begin{enumerate}
\item [C$2$.] each column has exactly $Z$ stars.
   \end{enumerate}

\begin{lemma}(\cite{YCTC})
\label{th-Fundamental}
An $F$-division coded caching scheme for $(K,M,N)$ caching system can be realized by a $(K,F,Z,S)$ PDA with memory fraction $\frac{M}{N}=\frac{Z}{F}$ and transmission rare $R=\frac{S}{F}$.
\end{lemma}

For the detailed realization method, the interested reader is referred to \cite{YCTC}.
In a PDA, a star not contained in any subarray showed as C$1$-b of Definition \ref{def-PDA}, is called useless. The authors in \cite{ZCWZ} pointed out that the useless stars not only make no contribution to reducing the transmission rate of a coded caching scheme realized by that PDA, but also result in a high subpacketization level and a large memory fraction. If each column of a $(K,F,Z,S)$ PDA has $Z'$ useless stars,  the authors in \cite{ZCWZ} improved the realization method in \cite{YCTC} by deleting all the useless stars and using an $[F,F-Z']_q$ maximum distance separable code for some prime power $q$, and came up with a new coded caching scheme with smaller transmission rate, memory fraction and subpacketization. That is the following result.
\begin{lemma}(\cite{ZCWZ})
\label{le-coded PDA}
For any $(K,F,Z,S)$ PDA， if there exist  $Z'$ useless stars in each column, then we can obtain $(F-Z')$-division $(K,M,N)$ coded caching scheme with memory fraction $\frac{M}{N}=\frac{Z-Z'}{F-Z'}$ and transmission rate $R=\frac{S}{F-Z'}$.
\end{lemma}

From Lemma \ref{th-Fundamental} and Lemma \ref{le-coded PDA}, we can obtain a coded caching scheme with small transmission rate and low subpacketization by constructing an appropriate PDA. Clearly given a $(K,F,Z,S)$ PDA, $\frac{Z}{F}>\frac{Z-Z'}{F-Z'}$ and $F>F-Z'$ always hold for any positive integer. So the scheme in Lemma \ref{le-coded PDA} has smaller memory fraction and subpacketization than that of the scheme from Lemma Lemma \ref{th-Fundamental}. However the operation field of the scheme from  Lemma \ref{le-coded PDA} is larger than or equal to the scheme from Lemma \ref{th-Fundamental}.

\section{New construction via Hamming distance}
\label{se-characterization}
In this section, we propose a new construction framework via Hamming distance to generate arrays which satisfy some of the conditions of PDA, and then obtain new PDAs through partitioning the entries of these arrays.

\subsection{The framework of constructing via Hamming distance}
\label{subsec-delivery strategy}
Let $\mathbf{x}$ and $\mathbf{y}$ be vectors of length $m$. The Hamming distance from $\mathbf{x}$ to $\mathbf{y}$, denoted by $d(\mathbf{x},\mathbf{y})$, is defined to be the number of coordinates at which $\mathbf{x}$ and $\mathbf{y}$ differ. The Hamming weight of $\mathbf{x}$, denoted by $wt(\mathbf{x})$, is defined to be the number of nonzero coordinates in $\mathbf{x}$.
\begin{construction}
\label{construction}
For any positive integers $m$, $\omega$, $F$, $K$ and $q\geq 2$ with $\omega<m$, given two subsets $\mathcal{A}$ and $\mathcal{B}$ of $[0,q)^m$ where $|\mathcal{A}|=F$ and $|\mathcal{B}|=K$, an $F\times K$ array $\mathbf{P}=\left(p_{\mathbf{a}, \mathbf{b}}\right), \mathbf{a}\in\mathcal{A}, \mathbf{b}\in \mathcal{B}$ is obtained as follows:
\begin{eqnarray}
\label{eq-PDA-P}
p_{\mathbf{a},\mathbf{b}}=\left\{
\begin{array}{cc}
        \mathbf{a}+\mathbf{b}&\ \ \ \ \ \ \hbox{if}\ d(\mathbf{a}, \mathbf{b})= \omega\\
        *&\ \ \hbox{otherwise}
       \end{array}\right.
\end{eqnarray}
where $\mathbf{a}=(a_0,\ldots, a_{m-1})$, $\mathbf{b}=(b_0,\ldots,b_{m-1})$. Here $\mathbf{a}\pm\mathbf{b}=(a_0\pm b_0,a_1\pm b_1, \ldots, a_{m-1}\pm b_{m-1})$ and all the operations are carried under mod $q$ in this paper.
\end{construction}
\begin{example}
\label{ex-ori}
When $m=3$, $\omega=2$, and $\mathcal{A}=\mathcal{B}=[0,2)^{m}$, the following $8\times 8$ array $\mathbf{P}$ can be obtained by Construction \ref{construction}.
\begin{eqnarray}\label{array-ex-ori}
\renewcommand\arraystretch{0.5}
\begin{array}{|c|cccccccc|}\hline
\mathbf{a} \backslash \mathbf{b}&
     000&100&010&110&001&101&011&111\\ \hline
000& \ast & \ast & \ast & 110 & \ast & 101 & 011 & \ast \\
100& \ast & \ast & 110 & \ast & 101 & \ast & \ast & 011 \\
010& \ast & 110 & \ast & \ast & 011 & \ast & \ast & 101 \\
110& 110 & \ast & \ast & \ast & \ast & 011 & 101 & \ast \\
001& \ast & 101 & 011 & \ast & \ast & \ast & \ast & 110 \\
101& 101 & \ast & \ast & 011 & \ast & \ast & 110 & \ast \\
011& 011 & \ast & \ast & 101 & \ast & 110 & \ast & \ast \\
111& \ast & 011 & 101 & \ast & 110 & \ast & \ast & \ast \\\hline
\end{array}
\end{eqnarray}In this paper, all the vectors in examples are always written as a string, e.g., $(1,1,0,0)$ is written as $1100$.
\end{example}
From Example \ref{ex-ori}, we can see that no vector occurs more than once in each row and each column. Furthermore when a vector ${\bf e}$ occurs in two distinct entries, i.e., $p_{{\bf a}_1,{\bf b}_1}=p_{{\bf a}_2,{\bf b}_2}={\bf e}$, we have $p_{{\bf a}_1,{\bf b}_2}=p_{{\bf a}_2,{\bf b}_1}=*$ if $d(\mathbf{a}_{1},\mathbf{b}_{2})\neq \omega$, otherwise $p_{{\bf a}_1,{\bf b}_2}\neq*$ and $p_{{\bf a}_2,{\bf b}_1}\neq*$. For instance, since vector $110$ occurs in entries $(110,000)$ and $(010,100)$ with $d(110,100)\neq2$, we have $p_{{\bf a}_1,{\bf b}_2}=p_{{\bf a}_1,{\bf b}_2}=*$ by \eqref{eq-PDA-P}. However since it also occurs in entries $(110,000)$ and $(011,101)$ with $d(110,101)=2$, we have $p_{{\bf a}_1,{\bf b}_2}=p_{{\bf a}_1,{\bf b}_2}\neq*$ by \eqref{eq-PDA-P}.
In fact this is not accidental. In general, we have the following proposition.

\begin{proposition}
\label{pro-property-integer}
Let $\mathbf{P}$ be the array generated by Construction \ref{construction}, if there are two distinct entries being the same vector, say $p_{\mathbf{a}_{1},\mathbf{b}_{1}}=p_{\mathbf{a}_{2},\mathbf{b}_{2}}=\mathbf{e}$, then the following two statements hold:
\begin{itemize}
\item [1)] The vector $\mathbf{e}$ occurs in different columns and different rows, i.e., the condition C$1$-a in Definition \ref{def-PDA} holds.
\item [2)] The subarray formed by rows $\mathbf{a}_{1}$, $\mathbf{a}_{2}$ and columns $\mathbf{b}_{1}$, $\mathbf{b}_{2}$ satisfies the condition C$1$-b in Definition \ref{def-PDA} if and only if $d(\mathbf{a}_{1},\mathbf{b}_{2})\neq\omega$.
\end{itemize}
\end{proposition}
\begin{proof} Suppose that a vector $\mathbf{e}=(e_0,e_1,\ldots,e_{m-1})\in [0,q)^m$ occurs in two distinct entries, say $(\mathbf{a}_{1}, \mathbf{b}_{1})$ and $(\mathbf{a}_{2}, \mathbf{b}_{2})$ where
\begin{eqnarray*}
&\mathbf{a}_{1}=(a_{1,0},a_{1,1},\ldots,a_{1,m-1}),&\
\mathbf{a}_{2}=(a_{2,0},a_{2,1},\ldots,a_{2,m-1}),\\
&\mathbf{b}_{1}=(b_{1,0},b_{1,1},\ldots,b_{1,m-1}),&\
\mathbf{b}_{2}=(b_{2,0},b_{2,1},\ldots,b_{2,m-1}).
\end{eqnarray*}
From Construction \ref{construction} we have
\begin{eqnarray}
\label{eq-aabb}
\mathbf{e}=\mathbf{a}_{1}+\mathbf{b}_{1}=\mathbf{a}_{2}+\mathbf{b}_{2}.
\end{eqnarray} That is, $\mathbf{a}_{1}-\mathbf{a}_{2}=\mathbf{b}_{2}-\mathbf{b}_{1}$. Clearly $\mathbf{a}_{1}=\mathbf{a}_{2}$ if and only if $\mathbf{b}_{1}=\mathbf{b}_{2}$. So vector $\mathbf{e}$ occurs in the different columns and different rows.
If $d(\mathbf{a}_{1},\mathbf{b}_{2})\neq\omega$, then $p_{\mathbf{a}_{1},\mathbf{b}_{2}}=*$ from Construction \ref{construction}. Since $d(\mathbf{a}_{1},\mathbf{b}_{2})=wt(\mathbf{a}_{1}-\mathbf{b}_{2})\neq\omega$, by \eqref{eq-aabb} we have $\mathbf{a}_{2}-\mathbf{b}_{1}=\mathbf{a}_{1}-\mathbf{b}_{2}$, so $d(\mathbf{a}_{2},\mathbf{b}_{1})=d(\mathbf{a}_{1},\mathbf{b}_{2})\neq\omega$ holds. This implies that $p_{\mathbf{a}_{2},\mathbf{b}_{1}}=*$. Conversely if $p_{\mathbf{a}_{1},\mathbf{b}_{2}}=p_{\mathbf{a}_{2},\mathbf{b}_{1}}=*$, we also have $d(\mathbf{a}_{1},\mathbf{b}_{2})\neq \omega$ similarly.
\end{proof}

From Proposition \ref{pro-property-integer}, the array $\mathbf{P}$ generated by Construction \ref{construction} has satisfied the Condition C1-a in Definition \ref{def-PDA}. In order to construct PDAs we only need to make any two distinct entries having the same vectors satisfy Proposition \ref{pro-property-integer}-2). So we should make a partition for each collection of entries having the same vector, thereby part it into several non-intersection subsets such that any two different entries in the same subset satisfy Proposition \ref{pro-property-integer}-2). That is the main discussion in the following subsection.

\subsection{The partitions of the entries in $\mathbf{P}$}
For any positive integers $m$, $q$, $\omega$ and the given subsets $\mathcal{A}$, $\mathcal{B}\in [0,q)^m$ with $\omega<m$ and $q\geq2$, we can obtain an array $\mathbf{P}$ by Construction \ref{construction}. Assume that vector $\mathbf{e}$ occurs $g_{\mathbf{e}}$ times in $\mathbf{P}$, say $p_{\mathbf{a}_1,{\bf b}_1}=\ldots=p_{\mathbf{a}_{g_{\mathbf{e}}},\mathbf{b}_{g_{\mathbf{e}}}}=\mathbf{e}$. Let $\mathcal{E}_{\mathbf{e}}=\{(\mathbf{a}_1,\mathbf{b}_1),(\mathbf{a}_2,\mathbf{b}_2),\ldots,(\mathbf{a}_{g_{\mathbf{e}}},{\bf b}_{g_{\mathbf{e}}})\}$. We claim that for some integer $1\leq h_{\mathbf{e}}\leq g_{{\bf e}}$ there always exists a partition $\mathcal{X}_{\mathbf{e}}=\{\mathcal{X}_{\mathbf{e},0}$, $\mathcal{X}_{\mathbf{e},1}$, $\ldots$, $\mathcal{X}_{\mathbf{e},h_{\mathbf{e}}-1}\}$ for $\mathcal{E}_{\mathbf{e}}$, 
satisfying \\[0.2cm]
\textbf{Property 1:} For any two different entries $(\mathbf{a}_{1},\mathbf{b}_{1})$, $(\mathbf{a}_{2},\mathbf{b}_{2})\in \mathcal{X}_{\mathbf{e},i}$, $i\in[0,h_{\mathbf{e}})$, $d(\mathbf{a}_{1},\mathbf{b}_{2})\neq\omega$ always holds.

Extremely, when $h_{\mathbf{e}}=g_{\mathbf{e}}$ there exists an trivial partition where $\mathcal{X}_{\mathbf{e},i}=\{(\mathbf{a}_{i},\mathbf{b}_{i})\}$ for $i\in h_{\mathbf{e}}$.
\begin{construction}
\label{construction-2}
Given an $F\times K$ array $\mathbf{P}$ generated by Construction \ref{construction}, we can obtain a new array $\mathbf{P}'=(p'_{\mathbf{a},\mathbf{b}})$, $\mathbf{a}\in \mathcal{A}$ and $\mathbf{b}\in \mathcal{B}$, where
\begin{eqnarray*}
p'_{\mathbf{a},\mathbf{b}}=\left\{\begin{array}{cc}
                  (\mathbf{e},i) & \ \ \hbox{if}\ p_{\mathbf{a},\mathbf{b}}=\mathbf{e},\ (\mathbf{a},\mathbf{b})\in \mathcal{X}_{\mathbf{e},i},i\in[0,h_{\mathbf{e}})\\
                  * &\ \ \ \hbox{otherwise}
                \end{array}
\right.
\end{eqnarray*}
 where $\mathcal{X}_{\mathbf{e}}=\{\mathcal{X}_{\mathbf{e},0}, \ldots, \mathcal{X}_{\mathbf{e},h_{\mathbf{e}-1}}\}$ is a partition satisfying \textbf{Property 1}. Clearly any two entries $p_{\mathbf{a}_{1},\mathbf{b}_{1}}=p_{\mathbf{a}_{2},\mathbf{b}_{2}}=(\mathbf{e},i)$ in $\mathbf{P}'$ satisfy $d(\mathbf{a}_{1},\mathbf{b}_{2})\neq\omega$. From Proposition \ref{pro-property-integer}, $\mathbf{P}'$ is a PDA.
\end{construction}
\begin{example}
\label{ex-ori2}
Let us consider the parameters in Example \ref{ex-ori} and the $8\times 8$ array in \eqref{array-ex-ori} again. For $\mathcal{E}_{110}$, $\mathcal{E}_{101}$, $\mathcal{E}_{011}$ we make their partitions satisfying \textbf{Property 1} as follow.
\begin{eqnarray*}
\mathcal{E}_{110}=\mathcal{X}_{110,0}\bigcup \mathcal{X}_{110,1}&=&\{(110,000),(000,110),(010,100),(100,010)\}\\
&&\bigcup\{(111,001),(001,111),(011,101),(101,011)\}
\end{eqnarray*}
\begin{eqnarray*}
\mathcal{E}_{101}=\mathcal{X}_{101,0}\bigcup \mathcal{X}_{101,1}&=&\{(101,000),(000,101),(001,100),(100,001)\}\\
&&\bigcup\{(111,010),(010,111),(011,110),(110,011)\}
\end{eqnarray*}
\begin{eqnarray*}
\mathcal{E}_{011}=\mathcal{X}_{011,0}\bigcup \mathcal{X}_{011,1}
&=&\{(011,000),(000,011),(001,010),(010,001)\}\\
&&\bigcup\{(111,100),(100,111),(101,110),(110,101)\}
\end{eqnarray*}

Based on the above partitions and Construction \ref{construction-2} a $(8,8,5,6)$ PDA $\mathbf{P}'$ can be obtained.
\begin{eqnarray}\label{array-ex-ori2}
\renewcommand\arraystretch{0.5}
\begin{array}{|c|cccccccc|}\hline
\mathbf{a} \backslash \mathbf{b}&
     000&100&010&110&001&101&011&111\\ \hline
000& \ast & \ast & \ast & 110,0 & \ast & 101,0 & 011,0 & \ast \\
100& \ast & \ast & 110,0 & \ast & 101,0 & \ast & \ast & 011,1 \\
010& \ast & 110,0 & \ast & \ast & 011,0 & \ast & \ast & 101,1 \\
110& 110,0 & \ast & \ast & \ast & \ast & 011,1 & 101,1 & \ast \\
001& \ast & 101,0 & 011,0 & \ast & \ast & \ast & \ast & 110,1 \\
101& 101,0 & \ast & \ast & 011,1 & \ast & \ast & 110,1 & \ast \\
011& 011,0 & \ast & \ast & 101,1 & \ast & 110,1 & \ast & \ast \\
111& \ast & 011,1 & 101,1 & \ast & 110,1 & \ast & \ast & \ast \\\hline
\end{array}
\end{eqnarray}
\end{example}

Generally, from Construction \ref{construction} and Construction \ref{construction-2} we have the following result.
\begin{theorem}
\label{th-main}
The $F\times K$ array $\mathbf{P}'$ generated from Construction \ref{construction-2} is a $(K,F,S)$ PDA with $S= \sum_{\mathbf{e}\in \mathbf{P}}h_{\mathbf{e}}$ where $h_{\mathbf{e}}$ is the cardinality of partition $\mathcal{X}_{\mathbf{e}}$ satisfying \textbf{Property 1}.
\end{theorem}

From the above introductions, we only need to consider designing the appropriate partition $\mathcal{X}_{\mathbf{e}}$ for each set $\mathcal{E}_\mathbf{e}$ in $\mathbf{P}$ generated by Construction \ref{construction}. In the following sections we will obtain several classes of PDAs for the parameter $q=2$ and $q=3$ by constructing partitions.

\section{The New schemes with $q=2$}
\label{con-q=2}
For any sets $\mathcal{A}$, $\mathcal{B}\subseteq [0,2)^m$ and all $\mathbf{e}$ in array $\mathbf{P}$ obtained by Construction \ref{construction}, when $m$ and $\omega$ are any integers with $\omega<m$ we first propose a partition $\mathcal{X}_{\mathbf{e}}$, and when $m\geq2\omega+1$ we improve this partition.
\subsection{The primary partition for any integer $m$ and $\omega$}
When $q=2$ for any vector $\mathbf{e}=(e_{0},e_{1},\ldots,e_{m-1})$ occurring in $\mathbf{P}$ obtained by Construction \ref{construction}, let $\mathcal{C}_{\mathbf{e}}=\{i\in[0,m)|e_{i}=0\}$. Clearly $|\mathcal{C}_{\mathbf{e}}|=m-\omega$ always holds.
\begin{partition}
\label{partition-q=2}
For any vector $\mathbf{e}$ in $\mathbf{P}$ obtained by Construction \ref{construction}, and for each vector $\mathbf{t}\in [0,2)^{m-\omega}$ we define
\begin{eqnarray}
\label{eq-par-1}
\mathcal{X}_{\mathbf{e},\mathbf{t}}=\{(\mathbf{a},\mathbf{b})|(\mathbf{a},\mathbf{b})\in \mathcal{E}_{\mathbf{e}}, \mathbf{a}|_{\mathcal{C}_{\mathbf{e}}}=\mathbf{b}|_{\mathcal{C}_{\mathbf{e}}}=\mathbf{t}\}.
\end{eqnarray}Here for any $m$ length vector $\mathbf{a}$ and a set $\mathcal{T}\subseteq [0,m)$, $\mathbf{a}|_{\mathcal{T}}$ is a vector obtained by deleting the coordinates $j\in[0,m)\setminus\mathcal{T}$. Let $\mathcal{X}_{\mathbf{e}}=\{\mathcal{X}_{\mathbf{e},\mathbf{t}}|\mathcal{X}_{\mathbf{e},\mathbf{t}}\neq \emptyset, \mathbf{t}\in [0,2)^{m-\omega}\}$. It is easy to check that $\mathcal{X}_{\mathbf{e}}$ is a partition for $\mathcal{E}_{\mathbf{e}}$.
\end{partition}

\begin{proposition}
\label{cor-q=2}Partition \ref{partition-q=2} satisfies \textbf{Property 1}. Furthermore, for any two different entries $(\mathbf{a}_{1},\mathbf{b}_{1})$ and $(\mathbf{a}_{2},\mathbf{b}_{2})$ in common partition of Partition \ref{partition-q=2}, $d(\mathbf{a}_{1},\mathbf{b}_{2})<\omega$ always holds.
\end{proposition}
\begin{proof}
For any vector $\mathbf{e}=(e_{0},e_{1},\ldots,e_{m-1})$ in $\mathbf{P}$, let us consider the $\mathcal{X}_{\mathbf{e}}$ in Partition \ref{partition-q=2}. For any $\mathcal{X}_{\mathbf{e},\mathbf{t}}\in \mathcal{X}_{\mathbf{e}}$, $\mathbf{t}\in [0,2)^{m-\omega}$, our statement always holds if $|\mathcal{X}_{\mathbf{e},\mathbf{t}}|=1$. If $|\mathcal{X}_{\mathbf{e},\mathbf{t}}|\geq 2$, for any two different entries $(\mathbf{a}_{1},\mathbf{b}_{1})$, $(\mathbf{a}_{2},\mathbf{b}_{2})\in \mathcal{X}_{\mathbf{e},\mathbf{t}}$, let
\begin{eqnarray*}
&\mathbf{a}_{1}=(a_{1,0},a_{1,1},\ldots,a_{1,m-1}),&\
\mathbf{a}_{2}=(a_{2,0},a_{2,1},\ldots,a_{2,m-1}),\\
&\mathbf{b}_{1}=(b_{1,0},b_{1,1},\ldots,b_{1,m-1}),&\
\mathbf{b}_{2}=(b_{2,0},b_{2,1},\ldots,b_{2,m-1}).
\end{eqnarray*}
For any $j\in \mathcal{C}_{\mathbf{e}}$ we have $a_{1,j}=b_{1,j}$, $a_{2,j}=b_{2,j}$ since $q=2$. By \eqref{eq-par-1} we have $$\mathbf{a}_{1}|_{\mathcal{C}_{\mathbf{e}}}=\mathbf{b}_{1}|_{\mathcal{C}_{\mathbf{e}}}=\mathbf{t}=\mathbf{a}_{2}|_{\mathcal{C}_{\mathbf{e}}}
=\mathbf{b}_{2}|_{\mathcal{C}_{\mathbf{e}}}$$
i.e., for any $j\in \mathcal{C}_{\mathbf{e}}$, $a_{1,j}=a_{2,j}$ and $a_{1,j}=b_{2,j}$ always hold. Due to $|\mathcal{C}_{\mathbf{e}}|=m-\omega$,  we have $d(\mathbf{a}_1,\mathbf{b}_2)\leq\omega$. From Proposition \ref{pro-property-integer}-1) we have $\mathbf{b}_{1}\neq \mathbf{b}_{2}$. So
$d(\mathbf{a}_1,\mathbf{b}_2)=d(\mathbf{a}_2,\mathbf{b}_1)<\omega$ always holds. Hence, $d(\mathbf{a}_{1},\mathbf{b}_{2})<\omega$ holds and partition $\mathcal{X}_{\mathbf{e}}$ satisfies the \textbf{Property 1}.
The proof is completed.
\end{proof}

We also consider the parameters in Example \ref{ex-ori} and the $8\times 8$ array in \eqref{array-ex-ori}. By Partition \ref{partition-q=2}, the partitions for $\mathcal{E}_{110}$, $\mathcal{E}_{101}$ and $\mathcal{E}_{011}$ are exactly showed in  Example \ref{ex-ori2}.
Then from Theorem \ref{th-main} an $(8,8,5,6)$ PDA can be obtained in \eqref{array-ex-ori2}. Furthermore, we can see that the stars in entries $(111,000)$, $(011,100)$, $(101,010)$, $(001,110)$, $(110,001)$, $(010,101)$, $(100,011)$ and $(000,111)$ are useless.
In fact, when $\mathcal{A}=\mathcal{B}=[0,2)^{m}$, the array generated by Construction \ref{construction-2} based on Partition \ref{partition-q=2} is the PDA where some stars in each column are useless. So from Theorem \ref{th-main} and Lemma \ref{le-coded PDA} the following result is obtained.

\begin{theorem}
\label{corollary-q=2-all}
For any positive integers $m$, $\omega$ with $\omega<m$, there exists a $(2^{m},2^{m},2^{m}-{m\choose \omega},{m\choose \omega}2^{m-\omega})$ PDA which can realize a $\sum^{\omega}_{i=0}{m\choose i}$-division $(2^{m},M,N)$ coded caching scheme with memory fraction $\frac{M}{N}=1-{m\choose \omega}/{\sum_{i=0}^{\omega}{m\choose i}}$ and transmission rate $R=\frac{{m\choose \omega}2^{m-\omega}}{\sum^{\omega}_{i=0}{m\choose i}}$.
\end{theorem}
\begin{proof}
Let $\mathcal{A}=\mathcal{B}=[0,2)^{m}$. From Construction \ref{construction} a $2^{m}\times2^{m}$ array $\mathbf{P}$ is obtained.
From \eqref{eq-PDA-P} the number of no-star entries in each column is ${m\choose \omega}$ so that $Z=2^{m}-{m\choose \omega}$. Since $\mathcal{A}=\mathcal{B}=[0,2)^{m}$, the collection of all vectors occurring in $\mathbf{P}$ is exactly the collection of all binary vectors with Hamming weight of $\omega$, i.e., the number of vectors occurring in $\mathbf{P}$ is $S'={m\choose \omega}$.
By Partition \ref{partition-q=2}, for any vector $\mathbf{e}$ we have
$h_{\mathbf{e}}=|\mathcal{X}_{\mathbf{e}}|=| \{\mathcal{X}_{\mathbf{e},\mathbf{t}} \ \left|\right.\ \mathcal{X}_{\mathbf{e},\mathbf{t}}\neq \emptyset\}|
=|\{\mathbf{t} \left|\right.\ \mathbf{a}\in \mathcal{A},\mathbf{b}\in \mathcal{B}, (\mathbf{a},\mathbf{b})\in \mathcal{E}_{\mathbf{e}}, \mathbf{a}|_{\mathcal{C}_{\mathbf{e}}}=\mathbf{b}|_{\mathcal{C}_{\mathbf{e}}}=\mathbf{t} \}|
=|[0,2)^{m-\omega}|=2^{m-\omega}$.
From Theorem \ref{th-main}, a $(2^{m},2^{m},2^{m}-{m\choose \omega},{m\choose \omega}2^{m-\omega})$ PDA is obtained where $S=S'h_{\mathbf{e}}={m\choose \omega}2^{m-\omega}$.
From Proposition \ref{cor-q=2}, for any two different entries $(\mathbf{a}_{1},\mathbf{b}_{1})$ and $(\mathbf{a}_{2},\mathbf{b}_{2})$ in common partition of Partition \ref{partition-q=2}, we have $p_{\mathbf{a}_{1},\mathbf{b}_{2}}=*$ and $d(\mathbf{a}_{1},\mathbf{b}_{2})<\omega$.
This means that for any entry $p_{\mathbf{a},\mathbf{b}}=*$ where $\mathbf{a}\in \mathcal{A}$, $\mathbf{b}\in \mathcal{B}$ with $d(\mathbf{a},\mathbf{b})>\omega$, the star isn't contained by any subarray showed as C$1$-b of Definition \ref{def-PDA}, i.e.,
it's useless star. Then the number of useless stars in each column is $Z'=\sum^{m}_{i=\omega+1}{m\choose i}$. From Lemma \ref{le-coded PDA}, the proof of Theorem \ref{corollary-q=2-all} is completed.
\end{proof}

In Lemma \ref{le-coded PDA}, the authors improved the coded caching scheme by deleting the useless stars of the PDA. This implies that we prefer to get a PDA with the number of useless stars as small as possible. So we can also reduce the number of useless stars by improving the Partition \ref{partition-q=2}.
\subsection{The improve partition for $m\geq2\omega+1$}
\label{subsec-q=2-improve-partition}
When $m\geq2\omega+1$ we further merge some elements of Partition \ref{partition-q=2} to obtain an improved partition with smaller cardinality. Then from Theorem \ref{th-main} we can get a new PDA with smaller value $S$ than that of the PDA in Theorem \ref{corollary-q=2-all}.
\begin{partition}
\label{partition-q=2'}
When $m\geq2\omega+1$, for any vector $\mathbf{e}$ in $\mathbf{P}$ obtained by Construction \ref{construction}, there always exists an $l_{\mathbf{e}}$-partition $\mathcal{D}=\{\mathcal{D}_{0}$ , $\mathcal{D}_{1}$, $\ldots$, $\mathcal{D}_{l_{\mathbf{e}-1}}\}$ for $[0,2)^{m-\omega}$ such that the distance of any two different vectors in $\mathcal{D}_{i}$, $i\in[0,l_{\mathbf{e}})$, is at least $\omega+1$.
For $\mathcal{X}_{\mathbf{e}}$ in Partition \ref{partition-q=2} and for each $\mathcal{D}_{i}$, let
\begin{equation*}
\mathcal{Y}_{\mathbf{e},\mathcal{D}_{i}}=\bigcup_{\mathbf{t}\in \mathcal{D}_{i}}\mathcal{X}_{\mathbf{e},\mathbf{t}}, \ \ \ \mathcal{X}_{\mathbf{e}, \mathbf{t}}\in\mathcal{X}_{\mathbf{e}}.
\end{equation*}
Then partition $\mathcal{Y}_{\mathbf{e}}=\{\mathcal{Y}_{\mathbf{e},\mathcal{D}_{0}}, \mathcal{Y}_{\mathbf{e},\mathcal{D}_{1}}, \cdots, \mathcal{Y}_{\mathbf{e},\mathcal{D}_{l_{\mathbf{e}}-1}}\}$ is an $l_{\mathbf{e}}$-partition for $\mathcal{E}_{\mathbf{e}}$ for some integer $l_{\bf e}$.
\end{partition}
\begin{proposition}
\label{pro-q=2-further}
The partition \ref{partition-q=2'} satisfies \textbf{Property 1}.
\end{proposition}

\begin{proof}
From Proposition \ref{cor-q=2}, it is sufficient to consider any two different entries $(\mathbf{a}_{1},\mathbf{b}_{1})\in \mathcal{X}_{\mathbf{e},\mathbf{t}_{{j_1}}}$ and $(\mathbf{a}_{2},\mathbf{b}_{2})\in \mathcal{X}_{\mathbf{e},\mathbf{t}_{{j_2}}}$, $\mathbf{t}_{{j_1}}$, $\mathbf{t}_{{j_2}}\in \mathcal{D}_{i}$, $i\in[0,l_{\mathbf{e}})$, ${j_1}\neq {j_2}$. By \eqref{eq-par-1} we have $\mathbf{t}_{{j_1}}=\mathbf{a}_{1}|_{\mathcal{C}_{\mathbf{e}}}$, $\mathbf{t}_{{j_2}}=\mathbf{b}_{2}|_{\mathcal{C}_{\mathbf{e}}}$. From Partition \ref{partition-q=2'} we have $d(\mathbf{a}_{1},\mathbf{b}_{2})\geq d(\mathbf{a}_{1}|_{\mathcal{C}_{\mathbf{e}}},\mathbf{b}_{2}|_{\mathcal{C}_{\mathbf{e}}})=d(\mathbf{t}_{{j_1}},\mathbf{t}_{{j_2}})\geq\omega+1$.
Clearly  $d(\mathbf{a}_{1},\mathbf{b}_{2})\neq \omega$ always holds. And the star $p_{\mathbf{a}_1,\mathbf{b}_2}=*$ and $p_{\mathbf{a}_2,\mathbf{b}_1}=*$ are both used.
\end{proof}

From Theorem \ref{th-main}, the PDA based on Partition \ref{partition-q=2'} has parameter $S= \sum_{\mathbf{e}\in \mathbf{P}}l_{\mathbf{e}}$. Given $\mathcal{A}$, $\mathcal{B}\subseteq[0,2)^m$ we only need to consider the value of $l_{\mathbf{e}}$ in partition $\mathcal{D}$ for each ${\bf e}$.
By means of the result on the vertex coloring, the following result can be obtained.
\begin{lemma}
\label{lem-upper'}
When $m\geq2\omega+1$ and $\mathcal{A}=\mathcal{B}=[0,2)^m$, there exists a $(2^{m},2^{m},2^{m}-{m\choose \omega},S)$ PDA where
\begin{itemize}
\item $S={m\choose \omega}2^{m-\omega-1}$ when $m=2\omega+1$.
\item $S\leq{m\choose \omega}\sum_{i=0}^{w}{m-\omega\choose i}$ when $m>2\omega+1$.
\end{itemize}
\end{lemma}
The proof is referred to Appendix A. We can see that the value of $S$ of the PDA in Theorem \ref{corollary-q=2-all} is larger than that of the PDA in Lemma \ref{lem-upper'} for the same $K=F=2^m$ and $Z=2^{m}-{m\choose \omega}$.
It is worth noting that there are also some stars of the PDA in Lemma \ref{lem-upper'} are useless. This implies that we can also get a coded caching scheme with small memory fraction from Lemma \ref{le-coded PDA}. Now let us take an example to verify our claim.
\begin{example}
When $m=4,\omega=1$ and $\mathcal{A}=\mathcal{B}=[0,2)^4$, from Construction \ref{construction}, Partition \ref{partition-q=2} and Construction \ref{construction-2}, a $(16,16,12,32)$ PDA can be obtained. For each vector $\mathbf{e}$ we have  $$\mathcal{X}_{e}=\{\mathcal{X}_{\mathbf{e},000}, \mathcal{X}_{\mathbf{e},100}, \mathcal{X}_{\mathbf{e},010}, \mathcal{X}_{\mathbf{e},110}, \mathcal{X}_{\mathbf{e},001}, \mathcal{X}_{\mathbf{e},101}, \mathcal{X}_{\mathbf{e},011}, \mathcal{X}_{\mathbf{e},111}\}.$$
Define a $4$-partition
$$\mathcal{D}=\{\mathcal{D}_{0}=\{000,111\}, \mathcal{D}_{1}=\{100,011\}, \mathcal{D}_{2}=\{010,101\},\mathcal{D}_{3}=\{110,001\}\}$$
for $[0,2)^{3}$ such that the distance of any two different vectors in one subset of $\mathcal{D}$ is exactly 3. Then from Partition \ref{partition-q=2'}, we have
\begin{eqnarray*}
\mathcal{Y}_{\mathbf{e}}&=&\{\mathcal{Y}_{\mathbf{e},\mathcal{D}_{0}}, \mathcal{Y}_{\mathbf{e},\mathcal{D}_{1}}, \mathcal{Y}_{\mathbf{e},\mathcal{D}_{2}}, \mathcal{Y}_{\mathbf{e},\mathcal{D}_{3}}\} \\
&=&\{\mathcal{X}_{\mathbf{e},000}\bigcup \mathcal{X}_{\mathbf{e},111}, \mathcal{X}_{\mathbf{e},100}\bigcup \mathcal{X}_{\mathbf{e},011}, \mathcal{X}_{\mathbf{e},010}\bigcup \mathcal{X}_{\mathbf{e},101},\mathcal{X}_{\mathbf{e},110}\bigcup \mathcal{X}_{\mathbf{e},001}\}.
\end{eqnarray*}
From Proposition \ref{pro-q=2-further} and Theorem \ref{th-main}, we can obtain a $(16,16,12,16)$ PDA listed in Table \ref{example-lemma2}. We can check that $S=16$ reaches the upper bound in Lemma \ref{lem-upper'} and each column has $Z'=6$ useless stars.
From Lemma \ref{le-coded PDA}, we can get a $10$-division $(16,M,N)$ coded caching scheme with memory fraction $\frac{M}{N}=\frac{3}{5}$ and transmission rate $R=\frac{6}{5}$.
\begin{table}[H]
\caption{The $(16,16,12,16)$ PDA based on Lemma \ref{le-coded PDA} }\label{example-lemma2}
    \centering
    \renewcommand\arraystretch{0.4}
    \setlength{\tabcolsep}{1.2mm}{
    \begin{tabular}
    {|c|c|c|c|c|c|c|c|c|c|c|c|c|c|c|c|c|}
    \hline
         $\mathbf{a}\backslash\mathbf{b}$& 0000 & 1000 & 0100 & 1100 & 0010 & 1010 & 0110 & 1110 & 0001 & 1001 & 0101 & 1101 & 0011 & 1011 & 0111 & 1111 \\ \hline

        0000 & * & \tabincell{c}{1000\\$\mathcal{D}_{0}$} & \tabincell{c}{0100\\$\mathcal{D}_{0}$} & * & \tabincell{c}{0010\\$\mathcal{D}_{0}$} & * & * & * & \tabincell{c}{0001\\$\mathcal{D}_{0}$} & * & * & * & * & * & * & * \\
        \hline
        1000 & \tabincell{c}{1000\\$\mathcal{D}_{0}$} & * & * & \tabincell{c}{0100\\$\mathcal{D}_{1}$} & * & \tabincell{c}{0010\\$\mathcal{D}_{1}$} & * & * & * & \tabincell{c}{0001\\$\mathcal{D}_{1}$} & * & * & * & * & * & * \\
        \hline
        0100 & \tabincell{c}{0100\\$\mathcal{D}_{0}$} & * & * & \tabincell{c}{1000\\$\mathcal{D}_{1}$} & * & * & \tabincell{c}{0010\\$\mathcal{D}_{2}$} & * & * & * & \tabincell{c}{0001\\$\mathcal{D}_{2}$} & * & * & * & * & * \\
        \hline
        1100 & * & \tabincell{c}{0100\\$\mathcal{D}_{1}$} & \tabincell{c}{1000\\$\mathcal{D}_{1}$} & * & * & * & * & \tabincell{c}{0010\\$\mathcal{D}_{3}$} & * & * & * & \tabincell{c}{0001\\$\mathcal{D}_{3}$} & * & * & * & * \\
        \hline
        0010 & \tabincell{c}{0010\\$\mathcal{D}_{0}$} & * & * & * & * & \tabincell{c}{1000\\$\mathcal{D}_{2}$} & \tabincell{c}{0100\\$\mathcal{D}_{2}$} & * & * & * & * & * & \tabincell{c}{0001\\$\mathcal{D}_{3}$} & * & * & * \\
        \hline
        1010 & * & \tabincell{c}{0010\\$\mathcal{D}_{1}$} & * & * & \tabincell{c}{1000\\$\mathcal{D}_{2}$} & * & * & \tabincell{c}{0100\\$\mathcal{D}_{3}$} & * & * & * & * & * & \tabincell{c}{0001\\$\mathcal{D}_{2}$} & * & * \\
        \hline
        0110 & * & * & \tabincell{c}{0010\\$\mathcal{D}_{2}$} & * & \tabincell{c}{0100\\$\mathcal{D}_{2}$} & * & * & \tabincell{c}{1000\\$\mathcal{D}_{3}$} & * & * & * & * & * & * & \tabincell{c}{0001\\$\mathcal{D}_{1}$} & * \\
        \hline
        1110 & * & * & * & \tabincell{c}{0010\\$\mathcal{D}_{3}$} & * & \tabincell{c}{0100\\$\mathcal{D}_{3}$} & \tabincell{c}{1000\\$\mathcal{D}_{3}$} & * & * & * & * & * & * & * & * & \tabincell{c}{0001\\$\mathcal{D}_{0}$} \\
        \hline
        0001 & \tabincell{c}{0001\\$\mathcal{D}_{0}$} & * & * & * & * & * & * & * & * & \tabincell{c}{1000\\$\mathcal{D}_{3}$} & \tabincell{c}{0100\\$\mathcal{D}_{3}$} & * & \tabincell{c}{0010\\$\mathcal{D}_{3}$} & * & * & * \\
        \hline
        1001 & * & \tabincell{c}{0001\\$\mathcal{D}_{1}$} & * & * & * & * & * & * & \tabincell{c}{1000\\$\mathcal{D}_{3}$} & * & * & \tabincell{c}{0100\\$\mathcal{D}_{2}$} & * & \tabincell{c}{0010\\$\mathcal{D}_{2}$} & * & * \\
        \hline
        0101 & * & * & \tabincell{c}{0001\\$\mathcal{D}_{2}$} & * & * & * & * & * & \tabincell{c}{0100\\$\mathcal{D}_{3}$} & * & * & \tabincell{c}{1000\\$\mathcal{D}_{2}$} & * & * & \tabincell{c}{0010\\$\mathcal{D}_{1}$} & * \\
        \hline
        1101 & * & * & * & \tabincell{c}{0001\\$\mathcal{D}_{3}$} & * & * & * & * & * & \tabincell{c}{0100\\$\mathcal{D}_{2}$} & \tabincell{c}{1000\\$\mathcal{D}_{2}$} & * & * & * & * & \tabincell{c}{0010\\$\mathcal{D}_{0}$} \\
        \hline
        0011 & * & * & * & * & \tabincell{c}{0001\\$\mathcal{D}_{3}$} & * & * & * & \tabincell{c}{0010\\$\mathcal{D}_{3}$} & * & * & * & * & \tabincell{c}{1000\\$\mathcal{D}_{1}$} & \tabincell{c}{0100\\$\mathcal{D}_{1}$} & * \\
        \hline
        1011 & * & * & * & * & * & \tabincell{c}{0001\\$\mathcal{D}_{2}$} & * & * & * & \tabincell{c}{0010\\$\mathcal{D}_{2}$} & * & * & \tabincell{c}{1000\\$\mathcal{D}_{1}$} & * & * & \tabincell{c}{0100\\$\mathcal{D}_{0}$} \\
        \hline
        0111 & * & * & * & * & * & * & \tabincell{c}{0001\\$\mathcal{D}_{1}$} & * & * & * & \tabincell{c}{0010\\$\mathcal{D}_{1}$} & * & \tabincell{c}{0100\\$\mathcal{D}_{1}$} & * & * & \tabincell{c}{1000\\$\mathcal{D}_{0}$} \\
        \hline
        1111 & * & * & * & * & * & * & * & \tabincell{c}{0001\\$\mathcal{D}_{0}$} & * & * & * & \tabincell{c}{0010\\$\mathcal{D}_{0}$} & * & \tabincell{c}{0100\\$\mathcal{D}_{0}$} & \tabincell{c}{1000\\$\mathcal{D}_{0}$} & * \\ \hline
    \end{tabular}
    }
\end{table}
\end{example}
Unfortunately due to nondeterminacy of Partition \ref{partition-q=2'} it is hard to propose an uniform function of the number of useless stars in each column of the PDA from Lemma \ref{lem-upper'} for any $m$ and $\omega$. By Lemma \ref{th-Fundamental}, the following result can be obtained.
\begin{theorem}
\label{th-fina-q=2}
For any positive integers $m$, $\omega$ with $m\geq2\omega+1$, there exists a $2^m$-division $(2^m,M,N)$ coded caching scheme with memory fraction $\frac{M}{N}=1-\frac{{m\choose \omega}}{2^m}$ and transmission rate
\begin{itemize}
\item $R=\frac{{m\choose \omega}}{2^{\omega+1}}$ if $m=2\omega+1$
\item $R\leq\frac{{m\choose \omega}\sum^{\omega}_{i=0}{m-\omega\choose i}}{2^m}$ if $m>2\omega+1$.
\end{itemize}
\end{theorem}

Finally we should point out that our proposed framework consisting Construction \ref{construction} and Construction \ref{construction-2} are also useful to construct the PDAs for any positive integer $q$. We will take $q=3$ as another example in the following section.
\section{New schemes with parameter $q=3$}
\label{con-q=3}
Similar to Section \ref{con-q=2}, in this section we also put forward two partitions for the case of $q=3$. That is, a primary partition for any integers $m$ and $\omega$ with $\omega<m$ and an improved partition when $m>\frac{3\omega}{2}$ are proposed.
\subsection{The primary partition for parameters $m$ and $\omega$}
When $q=3$, define $\mathcal{C}_{\mathbf{a}-\mathbf{b}}=\{i\in[0,m)\left|\right.a_i=b_i\}$ for each element $(\mathbf{a},\mathbf{b})$ in $\mathcal{E}_{\mathbf{e}}$. Then $|\mathcal{C}_{\mathbf{a}-\mathbf{b}}|=m-\omega$ always holds.
\begin{partition}
\label{partition-q=3}
For any vector $\mathbf{e}$ in $\mathbf{P}$ obtained by Construction \ref{construction} and for each element $\mathcal{T}\in {[0,m)\choose m-\omega}$ $=\{\mathcal{T}\ |\   \mathcal{T}\subseteq [0,m), |\mathcal{T}|=t\}$, define
\begin{eqnarray}
\label{eq-par-2}
\mathcal{X}_{\mathbf{e},\mathcal{T}}=\{(\mathbf{a},\mathbf{b})\ |\ (\mathbf{a},\mathbf{b})\in \mathcal{E}_{\mathbf{e}}, \mathcal{C}_{\mathbf{a}-\mathbf{b}}=\mathcal{T}\}.
\end{eqnarray}
Let $\mathcal{X}_{\mathbf{e}}=\{\mathcal{X}_{\mathbf{e},\mathcal{T}}\ |\ \mathcal{X}_{\mathbf{e},\mathcal{T}}\neq \emptyset, \mathcal{T}\in {[0,m)\choose m-\omega}\}$. Then $\mathcal{X}_{\mathbf{e}}$ is a partition of $\mathcal{E}_{\mathbf{e}}$.
\end{partition}

\begin{proposition}
\label{cor-q=3}
Partition \ref{partition-q=3} satisfies \textbf{Property 1}. Furthermore, for any two different entries $(\mathbf{a}_{1},\mathbf{b}_{1})$ and $(\mathbf{a}_{2},\mathbf{b}_{2})$ in common partition of Partition \ref{partition-q=3}, $d(\mathbf{a}_{1},\mathbf{b}_{2})<\omega$ always holds.
\end{proposition}

\begin{proof}
In Partition \ref{partition-q=3}, for any $\mathcal{X}_{\mathbf{e},\mathcal{T}}\in \mathcal{X}_{\mathbf{e}}$, $\mathcal{T}\in {[0,m)\choose m-\omega}$, our statement always holds if $|\mathcal{X}_{\mathbf{e},\mathcal{T}}|=1$.
If $|\mathcal{X}_{\mathbf{e},\mathcal{T}}|\geq 2$, let us consider any two different entries $(\mathbf{a}_{1},\mathbf{b}_{1})$, $(\mathbf{a}_{2},\mathbf{b}_{2})\in \mathcal{X}_{\mathbf{e},\mathcal{T}}$. By \eqref{eq-aabb}, we have $\mathbf{a}_{1}|_{ \mathcal{T}}+\mathbf{b}_{1}|_{ \mathcal{T}}=\mathbf{a}_{2}|_{ \mathcal{T}}+\mathbf{b}_{2}|_{ \mathcal{T}}$. In addition, by \eqref{eq-par-2},
$\mathbf{a}_{1}|_{ \mathcal{T}}=\mathbf{b}_{1}|_{ \mathcal{T}}$ and
$\mathbf{a}_{2}|_{ \mathcal{T}}=\mathbf{b}_{2}|_{ \mathcal{T}}$ always hold. Then we have
\begin{eqnarray}
\label{eq-a1=a2}
\mathbf{a}_{1}|_{ \mathcal{T}}=\mathbf{b}_{1}|_{ \mathcal{T}}=\mathbf{a}_{2}|_{ \mathcal{T}}=\mathbf{b}_{2}|_{ \mathcal{T}}
\end{eqnarray} since all the operations are under mod $q=3$.
For any $j\in[0,m)\setminus\mathcal{T}$
 \begin{eqnarray}
 \label{limit}
   a_{1,j}\neq b_{1,j}, \ \
   a_{2,j}\neq b_{2,j},
 \end{eqnarray}
holds due to $d(\mathbf{a}_{1},\mathbf{b}_{1})=d(\mathbf{a}_{2},\mathbf{b}_{2})=\omega$. Then by \eqref{eq-aabb} and \eqref{limit} we can get $a_{1,j}=b_{2,j}$ or $a_{1,j}=a_{2,j}$. If $a_{1,j}=a_{2,j}$ holds for each $j\in[0,m)\setminus\mathcal{T}$, we have $\mathbf{a}_{1}=\mathbf{a}_{2}$ by \eqref{eq-a1=a2}. Then $\mathbf{b}_{1}=\mathbf{b}_{2}$ by \eqref{eq-aabb} which contradicts our hypothesis $(\mathbf{a}_{1},\mathbf{b}_{1})\neq(\mathbf{a}_{2},\mathbf{b}_{2})$. So there are at least one coordinate, say $j'\in[0,m)\setminus\mathcal{T}$, such that $a_{1,j'}=b_{2,j'}$ holds. Then we have $d(\mathbf{a}_{1},\mathbf{b}_{2})<\omega$. So the partition $\mathcal{X}_{\mathbf{e}}$ satisfies the {\bf Property 1}.
\end{proof}

\begin{example}
When $m=3$, $\omega=2$, $\mathcal{A}=\mathcal{B}=[0,3)^{m}$, for any vector $\mathbf{e}$ in $\mathbf{P}$ obtained by Construction \ref{construction}, we have a $3$-partition $\mathcal{X}_{\mathbf{e}}=\{\mathcal{X}_{\mathbf{e},0}, \mathcal{X}_{\mathbf{e},1},\mathcal{X}_{\mathbf{e},2}\}$ from Partition \ref{partition-q=3}.
From Theorem \ref{th-main} a $(27,27,15,81)$ PDA $\mathbf{P}'$ can be obtained. Let us consider vector $\mathbf{e}=110$. The rows and the columns from $\mathcal{E}_{110}$ in $\mathbf{P}'$ form the following $12\times12$ subarray.
\begin{eqnarray*}
\begin{small}    \renewcommand\arraystretch{0.5}
\begin{array}{|c|cccccccccccc|}\hline
    \mathbf{a}\backslash\mathbf{b} & 000&100&010&110&201&211&202&212&021&121&022&122 \\ \hline
    000 & \ast & \ast & \ast & 110,2 & 201,1 & \ast & 202,1& \ast & 021,0 & \ast & 022,0 & \ast \\
    100 & \ast & \ast & 110,2 & \ast & 001,1 & \ast & 002,1 & \ast & \ast & 221,0 & \ast & 222,0 \\
    010 & \ast & 110,2 & \ast & \ast & \ast & 221,1 & \ast & 222,1 & 001,0 & \ast & 002,0 & \ast \\
    110 & 110,2 & \ast & \ast & \ast & \ast & 021,1 & \ast & 022,1 & \ast & 201,0 & \ast & 202,0 \\
    201 & 201,1 & 001,1 & \ast & \ast & \ast & \ast & \ast & 110,0& 222,2 & 022,2 & \ast & \ast \\
    211 & \ast & \ast & 221,1 & 021,1 & \ast & \ast & 110,0 & \ast & 201,2 & 002,2 & \ast & \ast \\
    202 & 202,1 & 002,1 & \ast & \ast & \ast & 110,0 & \ast & \ast & \ast & \ast & 221,2 & 021,2 \\
    212 & \ast & \ast & 222,1 & 022,1 & 110,0 & \ast & \ast & \ast & \ast & \ast & 201,2 & 001,2 \\
    021 & 021,0 & \ast & 001,0 & \ast & 222,2 & 202,2 & \ast & \ast & \ast & \ast & \ast & 110,1 \\
    121 & \ast & 221,0 & \ast & 201,0 & 022,2 & 002,2 & \ast & \ast & \ast & \ast & 110,1 & \ast \\
    022 & 022,0 & \ast & 000,0 & \ast & \ast & \ast & 221,2 & 201,2 & \ast & 110,1 & \ast & \ast \\
    122 & \ast & 222,0 & \ast & 202,0 & \ast & \ast & 021,2 & 001,2 & 110,1 & \ast & \ast & \ast \\
    \hline
\end{array}\end{small}
\end{eqnarray*}
\end{example}

From Proposition \ref{cor-q=3}, there are also some useless stars in PDA obtained by Partition \ref{partition-q=3}. Similar to Theorem \ref{corollary-q=2-all}, from Theorem \ref{th-main} and Lemma \ref{le-coded PDA} the following result can be obtained.

\begin{theorem}
\label{corollary-q=3-all}
For any positive integers $m$, $\omega$ with $\omega<m$, there exists a $(3^m,3^m,3^m-{m\choose \omega}2^{\omega},{m\choose \omega}3^m)$ PDA which gives a $\sum_{i=0}^{\omega}{m\choose i}2^i$-division $(3^m,M,N)$ coded caching scheme with memory fraction $\frac{M}{N}=1-\frac{{m\choose\omega}2^\omega}{\sum_{i=0}^{\omega}{m\choose i}2^i}$ and transmission rate $R=\frac{{m\choose \omega}3^m}{\sum_{i=0}^{\omega}{m\choose i}2^i}$.
\end{theorem}
\begin{proof}
Let $\mathcal{A}=\mathcal{B}=[0,3)^{m}$. From Construction \ref{construction} we can obtain a $3^{m}\times3^{m}$ array $\mathbf{P}$.
From \eqref{eq-PDA-P} the number of no-star entries in each column is ${m\choose \omega}2^{\omega}$ so that $Z=3^m-{m\choose \omega}2^{\omega}$.
Since for any vector $\mathbf{e}\in[0,3)^{m}$ there always exists vectors $\mathbf{a}\in \mathcal{A}$ and $\mathbf{b}\in \mathcal{B}$ where $\mathbf{e}=\mathbf{a}+\mathbf{b}$, $d(\mathbf{a},\mathbf{b})=\omega$, the number of vectors occurring in $\mathbf{P}$ is $S'=3^m$.
By Partition \ref{partition-q=3}, for any vector $\mathbf{e}$, we have
$h_{\mathbf{e}}=|\mathcal{X}_{\mathbf{e}}|=|\{\mathcal{T}\in {[0,m)\choose m-\omega}\left|\right. \mathcal{X}_{\mathbf{e},\mathcal{T}}\neq \emptyset\}|
=\left|\left\{\mathcal{C}_{\mathbf{a}-\mathbf{b}}
\ \left|\right.\ \mathbf{a}\in \mathcal{A},\mathbf{b}\in \mathcal{B},\ (\mathbf{a},\mathbf{b})\in \mathcal{E}_{\mathbf{e}}\right\}\right|
={m\choose \omega}$. From Theorem \ref{th-main}, we have a $(3^m,3^m,3^m-{m\choose \omega}2^{\omega},{m\choose \omega}3^m)$ PDA where $S=S'h_{\mathbf{e}}={m\choose \omega}3^m$.

From Proposition \ref{cor-q=3}, for any two different entries $(\mathbf{a}_{1},\mathbf{b}_{1})$ and $(\mathbf{a}_{2},\mathbf{b}_{2})$ in common partition of Partition \ref{partition-q=2}, we have $p_{\mathbf{a}_{1},\mathbf{b}_{2}}=*$ and $d(\mathbf{a}_{1},\mathbf{b}_{2})<\omega$.
This means that for any entry $p_{\mathbf{a},\mathbf{b}}=*$ where $\mathbf{a}\in \mathcal{A}$, $\mathbf{b}\in \mathcal{B}$ with $d(\mathbf{a},\mathbf{b})>\omega$, the star isn't contained by any subarray showed as C$1$-b of Definition \ref{def-PDA}, i.e.,
it's useless star. Then the number of useless stars in each column is $Z'=\sum_{i=\omega+1}^{m}{m\choose i}2^i$. From Lemma \ref{le-coded PDA}, the proof of Theorem \ref{corollary-q=3-all} is completed.
\end{proof}

We can reduce the number of useless stars by modifying the Partition \ref{partition-q=3} to improve coded caching schemes.
\subsection{The improved partition when $m>\frac{3\omega}{2}$}
When $m>\frac{3\omega}{2}$, we further merge some elements of Partition \ref{partition-q=3} to obtain an improved partition with smaller cardinality.
\begin{partition}
\label{partition-q=3'}
When $m>\frac{3\omega}{2}$, for any vector $\mathbf{e}$ in $\mathbf{P}$ obtained by Construction \ref{construction}, there always exists an $l_{\mathbf{e}}$-partition $\mathcal{D}=\{\mathcal{D}_{0}$ , $\mathcal{D}_{1}$, $\ldots$, $\mathcal{D}_{l_{\mathbf{e}}-1}\}$ for ${[0,m)\choose m-\omega}$, such that the cardinality of the intersection of any two elements from $ \mathcal{D}_{i}$, $i\in[0,l_{\mathbf{e}})$ is less than $m-\frac{3\omega}{2}$.
For $\mathcal{X}_{\mathbf{e}}$ in Partition \ref{partition-q=3} and for each $\mathcal{D}_{i}$, let
\begin{equation*}
\mathcal{Y}_{\mathbf{e},\mathcal{D}_{i}}=\bigcup_{\mathcal{T}\in \mathcal{D}_{i}}\mathcal{X}_{\mathbf{e},\mathcal{T}}, \ \ \ \mathcal{X}_{\mathbf{e}, \mathcal{T}}\in\mathcal{X}_{\mathbf{e}}.
\end{equation*}
Then $\mathcal{Y}_{\mathbf{e}}=\{\mathcal{Y}_{\mathbf{e},\mathcal{D}_{0}}, \mathcal{Y}_{\mathbf{e},\mathcal{D}_{1}}, \cdots, \mathcal{Y}_{\mathbf{e},\mathcal{D}_{l_{\mathbf{e}}-1}}\}$ is an $l_{\mathbf{e}}$-partition for $\mathcal{E}_{\mathbf{e}}$.
\end{partition}

\begin{proposition}
\label{cor-q=3'}
Partition \ref{partition-q=3'} satisfies \textbf{Property 1}.
\end{proposition}
\begin{proof}
From Proposition \ref{cor-q=2}, we only need to consider any two different entries $(\mathbf{a}_{1},\mathbf{b}_{1})\in \mathcal{X}_{\mathbf{e},\mathcal{T}_{{j_1}}}$ and $(\mathbf{a}_{2},\mathbf{b}_{2})\in \mathcal{X}_{\mathbf{e},\mathcal{T}_{{j_2}}}$, $\mathcal{T}_{{j_1}}$, $\mathcal{T}_{{j_2}}\in \mathcal{D}_{i}$, $i\in[0,l_{\mathbf{e}})$, ${j_1}\neq {j_2}$.
We have
$\mathcal{T}_{{j_1}}=\mathcal{C}_{\mathbf{a}_{1}-\mathbf{b}_{1}}$ and $\mathcal{T}_{{j_2}}=\mathcal{C}_{\mathbf{a}_{2}-\mathbf{b}_{2}}$ due to Partition \ref{partition-q=3}. From Partition \ref{partition-q=3'} we have $|\mathcal{C}_{\mathbf{a}_{1}-\mathbf{b}_{1}}\bigcap\mathcal{C}_{\mathbf{a}_{2}-\mathbf{b}_{2}}|=|\mathcal{T}_{{j_1}}\bigcap \mathcal{T}_{{j_2}}|<m-\frac{3\omega}{2}$.
Since \eqref{eq-aabb}
and that for any $s\in \mathcal{C}_{\mathbf{a}_{1}-\mathbf{b}_{1}}\setminus(\mathcal{C}_{\mathbf{a}_{1}-\mathbf{b}_{1}}\bigcap\mathcal{C}_{\mathbf{a}_{2}-\mathbf{b}_{2}})$,
$a_{1,s}=b_{1,s}$ and $a_{2,s}\neq b_{2,s}$ always hold, we have $2a_{1,s}=a_{2,s}+b_{2,s}$. Hence $a_{1,s}\neq b_{2,s}$.
And for any $s\in \mathcal{C}_{\mathbf{a}_{2}-\mathbf{b}_{2}}\setminus(\mathcal{C}_{\mathbf{a}_{1}-\mathbf{b}_{1}}\bigcap\mathcal{C}_{\mathbf{a}_{2}-\mathbf{b}_{2}})$,
$a_{2,s}=b_{2,s}$ and $a_{1,s}\neq b_{1,s}$ always hold, we have $2a_{2,s}=a_{1,s}+b_{1,s}$. Hence $a_{1,s}\neq b_{2,s}$.
 So $d(\mathbf{a}_{1},\mathbf{b}_{2})\geq2(m-\omega-|\mathcal{C}_{\mathbf{a}_{1}-\mathbf{b}_{1}}\bigcap\mathcal{C}_{\mathbf{a}_{2}-\mathbf{b}_{2}}|)=2(m-\omega-|\mathcal{T}_{{j_1}}\bigcap \mathcal{T}_{{j_2}}|)>\omega$. Hence $d(\mathbf{a}_{1},\mathbf{b}_{2})\neq\omega$ always holds. Then $\mathcal{Y}_{\mathbf{e}}$ satisfies the {\bf Property 1}.
The proof is completed.
\end{proof}

Similar to the proof of Lemma \ref{lem-upper'}, using vertex coloring the following result can be obtained.

\begin{lemma}
\label{upper}
When $m>\frac{3\omega}{2}$ and $\mathcal{A}=\mathcal{B}=[0,3)^m$, there exists a $(3^m,3^m,3^m-{m\choose \omega}2^{\omega},S)$ where
$$S\leq3^m\cdot\left(1+\sum^{m-\omega-1}_{ i=\lceil m-\frac{3}{2}\omega\rceil}{m-\omega\choose i}{\omega\choose m-\omega-i}\right).$$
\end{lemma}
The proof of Lemma \ref{upper} is presented in Appendix B. Similar to the PDA in Lemma \ref{lem-upper'}, there are also some useless stars in the PDA from Lemma \ref{upper}. We can not find the exactly number of useless stars either. Applying Lemma \ref{th-Fundamental} to the PDA in Lemma \ref{upper}, the following result can be obtained.
\begin{theorem}
\label{co-repocess-PDA}
For any positive integers $m$, $\omega$ with $m>\frac{3\omega}{2}$, there exists a $3^m$-division $(3^m,M,N)$ coded caching scheme with memory fraction $\frac{M}{N}=1-\frac{{m\choose \omega}2^{\omega}}{3^{m}}$ and transmission rate $R\leq1+\sum^{m-\omega-1}_{ i=\lceil m-\frac{3}{2}\omega\rceil}{m-\omega\choose i}{\omega\choose m-\omega-i}$.
\end{theorem}

From Section \ref{con-q=2} and Section \ref{con-q=3}, we use the same method to construct the two partitions for $q=2$, $3$ respectively. However the partition for $q=3$ is more complicated than the partition for $q=2$. In fact, the complexity of constructing partition increases with the growth of $q$.
\section{The performance analysis of new schemes}
\label{comparison}
In this section several comparisons are proposed. First we claim that our schemes from Theorem \ref{corollary-q=2-all} and Theorem \ref{corollary-q=3-all} can achieve small memory fractions.
Now we take the memory fraction $\frac{M}{N}=1-\frac{{m\choose \omega}}{{\sum_{i=0}^{\omega}{m\choose i}}}$ of the scheme from Theorem 2 as an example. We have
\begin{eqnarray*}
\frac{M}{N}&=&1-\frac{{m\choose \omega}}{{\sum_{i=0}^{\omega}{m\choose i}}}=1-\frac{1}{\frac{{m\choose 0}}{{m\choose \omega}}+\frac{{m\choose 1}}{{m\choose \omega}}+\cdots+\frac{{m\choose \omega-1}}{{m\choose \omega}}+\frac{{m\choose \omega}}{{m\choose \omega}}}\\
&=&1-\frac{1}{\frac{\omega}{m-\omega+1}\frac{\omega-1}{m-\omega+2}\cdots\frac{1}{m}
+\frac{\omega}{m-\omega+1}\frac{\omega-1}{m-\omega+2}\ldots\frac{2}{m-1}+\cdots+
+\frac{\omega}{m-\omega+1}+1}\\
&<&1-\frac{1}{(\frac{\omega}{m-\omega+1})^{\omega}+(\frac{\omega}{m-\omega+1})^{\omega-1}+\cdots+
+\frac{\omega}{m-\omega+1}+1}\\
&=&1-\frac{1-\frac{\omega}{m-\omega+1}}{1-\left(\frac{\omega}{m-\omega+1}\right)^{\omega+1}}
\end{eqnarray*}

For any $\lambda\in[0,0.5)$ let $\omega=\lambda m$. Clearly
\begin{eqnarray*}
&&\lim_{m\rightarrow\infty}\frac{1-\frac{\omega}{m-\omega+1}}{1-\left(\frac{\omega}{m-\omega+1}\right)^{\omega+1}}
=\lim_{m\rightarrow\infty}\frac{1-\frac{\lambda m}{(1-\lambda)m+1}}{1-\left(\frac{\lambda m}{(1-\lambda)m+1}\right)^{\lambda m+1}}\\
&=&\lim_{m\rightarrow\infty}\frac{1-\frac{\lambda }{1-\lambda}}{1-\left(\frac{\lambda }{1-\lambda}\right)^{\lambda m+1}}\\
&=&1-\frac{\lambda}{1-\lambda}=\frac{1-2\lambda}{1-\lambda}
\end{eqnarray*}
So we have
\begin{eqnarray*}
\lim_{m\rightarrow\infty}\frac{M}{N}=1-\frac{1-2\lambda}{1-\lambda}=\frac{\lambda}{1-\lambda}
\end{eqnarray*}
If $\lambda< \frac{1}{3}$, the memory fraction $\frac{M}{N}$ of Theorem 2 is less than $\frac{1}{2}$ even when $m$ is large. Similarly we can also show that the scheme from Theorem 4 can achieve small memory fraction too.

\subsection{The comparison between new scheme in Theorem \ref{corollary-q=2-all} and MN Scheme in \cite{MN}}
From Theorem \ref{corollary-q=2-all} we have a coded caching scheme with
\begin{eqnarray*}
K=2^m, \ \ \frac{M}{N}=1-\frac{{m\choose \omega}}{{\sum_{i=0}^{\omega}{m\choose i}}},\ \ \ F_1=\sum_{i=0}^{\omega}{m\choose i}, \ \ \ \ R_1=\frac{{{m\choose \omega}2^{m-\omega}}}{\sum_{i=0}^{\omega}{m\choose i}}.
\end{eqnarray*}

When $t=2^m-\frac{2^m{m\choose \omega}}{\sum_{i=0}^{\omega}{m\choose i}}$, from the first row of Table \ref{tab-known-1}, we have an MN scheme with the $K=2^m$ and $\frac{M}{N}=1-\frac{{m\choose \omega}}{{\sum_{i=0}^{\omega}{m\choose i}}}$ where the subpacketization and transmission rate are respectively
\begin{eqnarray*}
F_{MN}={2^m \choose 2^m{m\choose \omega}\big/\sum_{i=0}^{\omega}{m\choose i}}, \ \ \ \ \
R_{MN}=\frac{2^m{m\choose \omega}}{\sum_{i=0}^{\omega}{m\choose i}\left(1+2^m\right)-2^m{m\choose \omega}}.
\end{eqnarray*}
Then we have the following ratios.
\begin{eqnarray}
\label{gene-ratio}
\frac{F_1}{F_{MN}}=\frac{\sum_{i=0}^{\omega}{m\choose i}}{{2^m \choose 2^m{m\choose \omega}\big/\sum_{i=0}^{\omega}{m\choose i}}}, \ \ \ \ \  \frac{R_1}{R_{MN}}=\frac{1+2^m-2^m{m\choose \omega}\big/\sum_{i=0}^{\omega}{m\choose i}}{2^\omega}
\end{eqnarray}
Since it's difficult to estimate an approximate ratio for any $m$ and $\omega$, we consider taking specific parameters that would lead to exact ratios. When $\omega=\frac{m}{2}$,
\begin{equation}\label{equa}
\sum_{i=0}^{\omega}{m\choose i}=\sum_{i=0}^{\frac{m}{2}}{m\choose i}=2^{m-1}+\frac{1}{2}{m\choose \frac{m}{2}}\approx 2^{m-1}\left(1+\frac{1}{\sqrt{\frac{\pi m}{2}}}\right)
\end{equation}
where ${m\choose \frac{m}{2}}\approx\frac{2^m}{\sqrt{\frac{\pi m}{2}}}$.
\eqref{gene-ratio} can be written as
\begin{eqnarray}
\frac{F_1}{F_{MN}}
\approx \frac{2^{m-1}\left(1+\frac{1}{\sqrt{\frac{\pi m}{2}}}\right)}
{{2^m\choose 2^{m+1}\big/\left(1+\sqrt{\frac{\pi m}{2}}\right)}}
&<\frac{2^{m-1}\left(1+\frac{1}{\sqrt{m}}\right)}
{{2^m \choose 2^m\big/\left(1+\sqrt{m}\right)}}\label{2} \\
&<\frac{2^{m-1}\left(1+\frac{1}{\sqrt{m}}\right)}{{\left(1+\sqrt{m}\right)}^{\frac{2^m}{1+\sqrt{m}}}}\label{3}\\
&\approx \frac{K}{2}\cdot \frac{1}{{(1+\sqrt{m})}^\frac{K}{1+\sqrt{m}}}\label{F-ratio}.
\end{eqnarray}by \eqref{equa},
where \eqref{2} is result of $\frac{2^{m+1}}{1+\sqrt{\frac{\pi m}{2}}}<\frac{2^{m+1}}{2}$ and $\frac{2^{m+1}}{1+\sqrt{\frac{\pi m}{2}}}>\frac{2^m}{1+\sqrt{m}}$, and \eqref{3} holds due to
\begin{align*}
{2^m \choose \frac{2^m}{1+\sqrt{m}}}&=\frac{(2^m)(2^m-1)(2^m-2)\cdots(\frac{2^m}{1+\sqrt{m}}+1)}{(\frac{2^m}{1+\sqrt{m}})(\frac{2^m}{1+\sqrt{m}}-1)(\frac{2^m}{1+\sqrt{m}}-2)\cdots1}>{(1+\sqrt{m})}^{\frac{2^m}{1+\sqrt{m}}}.
\end{align*}
Meanwhile,
\begin{eqnarray}\label{R-ratio}
\frac{R_1}{R_{MN}}&\approx\left(\frac{1}{2^m}+1-\frac{2}{1+\sqrt{\frac{\pi m}{2}}}\right)\cdot 2^{\frac{m}{2}}
\approx \sqrt{K}\cdot\left(1-\frac{1}{1+\sqrt{m}}\right)
\end{eqnarray}

We can see that for the same number of users and memory fraction, subpacketization $F_1$ of our scheme is at least $\frac{K}{2}\cdot \frac{1}{{(1+\sqrt{m})}^\frac{K}{1+\sqrt{m}}}$ times smaller than $F_{MN}$ of MN scheme from \eqref{F-ratio}, meanwhile the transmission rate $R_1$ is at most $\sqrt{K}\cdot\left(1-\frac{1}{1+\sqrt{m}}\right)$ times larger than $R_{MN}$ from \eqref{R-ratio}. Finally we propose the following example to further verify our claim.
\begin{example}
\label{specific-comp}
When $m=4,6,8$ and $10$, let $\omega=2,3,4$ and $5$ respectively. For the same number of users $K$ and memory fraction $\frac{M}{N}$, the values of $\frac{F_1}{F_{MN}}$ in \eqref{2} and the values of $\frac{R_1}{R_{MN}}$ in \eqref{R-ratio} are listed in Table \ref{tab-3}. Clearly the amount of the reducing subpacketizations is much larger than the amount of the increasing transmission rate.
\begin{table}[H]
\center    \renewcommand\arraystretch{0.5}
\caption{The comparison between scheme in Theorem \ref{corollary-q=2-all} and MN Scheme in \cite{MN}\label{tab-3}}
\begin{tabular}{|c|c|c|c|c|c|}
\hline
  $m$ & $\omega$ & $K$ & $\frac{M}{N}$ & $\frac{F_1}{F_{MN}}$ & $\frac{R_1}{R_{MN}}$ \\
  \hline
  $4$ & $2$ & $2^4$ & $0.545$ & $9\times 10^{-4}$ & $1.940$ \\
  $6$ & $3$ & $2^6$ & $0.531$ & $2.625\times 10^{-19}$ & $4.445$ \\
  $8$ & $4$ & $2^8$ & $0.570$ & $3.546\times10^{-48}$ & $9.186$ \\
  $10$ & $5$ & $2^{10}$ & $0.605$ & $2.178\times10^{-298}$ & $19.415$ \\
  \hline
\end{tabular}
\end{table}
\end{example}
\subsection{The comparison between new schemes and knowing linear subpacketization schemes}
In this section, we further discuss the performance of our schemes from Theorem \ref{corollary-q=2-all} and Theorem \ref{corollary-q=3-all} by comparing known schemes with linear subpacketizations proposed in \cite{CK} and \cite{YTCC} respectively.
Unfortunately, the theoretical analyses of the comparisons become quite messy and do not yield much intuition. Instead, we illustrate the advantages on the parameters of the user number, memory fraction, subpacketization and transmission rate by numerical comparisons. For the sake of clarity, we mark the parameters of scheme in \cite{CK} as $(k,n,m,t,q)$, the parameters of scheme in \cite{YTCC} as $(m,a,b,\lambda)$, the parameters of new scheme from Theorem \ref{corollary-q=2-all} as $(m,\omega,2)$, and the parameters of our new scheme from Theorem \ref{corollary-q=3-all} as $(m,\omega,3)$. In the following we always use the above parameters to denote their related scheme. For instance, parameters $(6,2,2,1,2)$ denotes the scheme in \cite{CK} with $k=6$, $n=m=2$, $t=1$ and $q=2$. Parameter $(10,3,3,2)$ denotes the scheme in \cite{YTCC} with $m=10,a=b=3$ and $\lambda=2$. Parameters $(16,6,2)$ denotes the scheme from Theorem \ref{corollary-q=2-all} with $m=16,\omega=6$ and $q=2$. And parameters $(10,6,3)$ denotes the scheme from Theorem \ref{corollary-q=3-all} with $m=10,\omega=3$ and $q=3$.

Firstly let us see the comparisons of our two new schemes from Theorem \ref{corollary-q=2-all} and Theorem \ref{corollary-q=3-all} and the scheme in \cite{CK} in Table \ref{com2}. We can see that our two new schemes both have lower subpacketizations, memory fractions and observably smaller transmission rate meanwhile can serve more users.

\begin{table}[H]
\center
\caption{The comparison between schemes in Theorem \ref{corollary-q=2-all}, Theorem \ref{corollary-q=3-all} and scheme in \cite{CK}}
\label{com2}
    \renewcommand\arraystretch{0.5}
\begin{tabular}{|c|c|c|c|c|c|}
\hline
 schemes & parameters & $K$ & $F$ & $\frac{M}{N}$ & $R$ \\
\hline
$(k,n,m,t,q)$ in \cite{CK} & $(6,2,2,1,2)$ & $39060$ & $39060$ & $0,6330$ & $716.8000$\\
$(m,\omega,2)$ in Theorem \ref{corollary-q=2-all} & $(16,6,2)$ & $65536$ & $14893$ & $0.4623$ & $550.6071$ \\
$(m,\omega,3)$ in Theorem \ref{corollary-q=3-all} & $(10,6,3)$ & $59049$ & $26025$ & $0.4836$ & $476.4761$ \\
\hline
$(k,n,m,t,q)$ in \cite{CK} & $(6,2,2,1,3)$ & $7927920$ & $7927920$ & $0.4190$ & $230291.1000$\\
$(m,\omega,2)$ in Theorem \ref{corollary-q=2-all} & $(23,7,2)$  & $8388608$ & $390656$ & $0.3724$ & $41127.2556$ \\
$(m,\omega,3)$ in Theorem \ref{corollary-q=3-all} & $(15,7,3)$ & $14348907$ & $1266027$ & $0.3494$ & $72933.0548$ \\
\hline
$(k,n,m,t,q)$ in \cite{CK} & $(6,2,2,1,4)$ & $422021600$ & $422021600$ & $0.3043$ & $14680064.0000$\\
 $(m,\omega,2)$ in Theorem \ref{corollary-q=2-all} & $(19,7,2)$ & $536870912$ & $2182396$ & $0.2848$ & $2999632.4210$\\
$(m,\omega,3)$ in Theorem \ref{corollary-q=3-all} & $(19,7,3)$  & $1162261467$ & $8628699$ & $0.2525$ & $6787121.7665$\\
\hline
$(k,n,m,t,q)$ in \cite{CK} & $(6,2,2,1,5)$ & $9914404500$ & $9914404500$ & $0.2366$ & $378417968.7500$\\
$(m,\omega,2)$ in Theorem \ref{corollary-q=2-all} & $(34,6,2)$  & $17179869184$ & $1676116$ & $0.1976$ & $215390771.5911$\\
$(m,\omega,3)$ in Theorem \ref{corollary-q=3-all} & $(21,7,3)$ & $10460353203$ & $10460353203$ & $0.2214$ & $63631596.7218$\\
\hline
$(k,n,m,t,q)$ in \cite{CK} & $(6,2,2,1,6)$ & $135288489420$ & $135288489420$ & $0.1928$ & $5460095692.8000$\\
$(m,\omega,2)$ in Theorem \ref{corollary-q=2-all} & $(37,5,2)$  & $137438953472$ &  $510416$ & $0.1460$ & $3667916678.6004$\\
$(m,\omega,3)$ in Theorem \ref{corollary-q=3-all} & $(24,6,3)$  & $282429536481$ &  $10161633$ & $0.1523$ & $3740922929.6312$\\
\hline
$(k,n,m,t,q)$ in \cite{CK} & $(6,2,2,1,7)$  & $1255883249600$ & $1255883249600$ & $0.1624$ & $52596891363.8000$\\
$(m,\omega,2)$ in Theorem \ref{corollary-q=2-all} & $(41,6,2)$  & $2199023255552$ & $5358578$ & $0.1609$ & $28831289808.7916$\\
$(m,\omega,3)$ in Theorem \ref{corollary-q=3-all} & $(26,6,3)$  & $2541865828329$ &  $17101033$ & $0.1384$ & $34220960199.0819$\\
\hline
\end{tabular}
\end{table}

Now let us compare our two new schemes from Theorem \ref{corollary-q=2-all} and Theorem \ref{corollary-q=3-all} with the scheme from \cite{YTCC} in Table \ref{com23}. We can see that our new scheme from Theorem \ref{corollary-q=2-all} has smaller or same subpacketization, lower memory fraction and smaller transmission rate meanwhile is able to serve more users. The our new scheme in Theorem \ref{corollary-q=3-all} has advantages on the user number, memory fraction and subpacketization at the cost of some transmission rate.
\begin{table}[H]
\center
\caption{The comparison between schemes in Theorem \ref{corollary-q=2-all}, Theorem \ref{corollary-q=3-all} and scheme in \cite{YTCC}}
\label{com23}    \renewcommand\arraystretch{0.5}
\begin{tabular}{|c|c|c|c|c|c|}
\hline
schemes & parameters & $K$ & $F$ & $\frac{M}{N}$ & $R$ \\
\hline
$(m,a,b,\lambda)$ in\cite{YTCC} & $(10,3,3,2)$ & $120$ & $120$ & $0.8250$ & $0.7500$\\
$(m,\omega,2)$ in Theorem \ref{corollary-q=2-all} & $(7,3,2)$ & $128$ & $120$ & $0.8250$ & $0.7000$ \\
\hline
$(m,a,b,\lambda)$ in\cite{YTCC} & $(16,4,5,2)$ & $1820$ & $4368$ & $0.6978$ & $10.0000$\\
$(m,\omega,2)$ in Theorem \ref{corollary-q=2-all} & $(11,6,2)$ & $2048$ & $1486$ & $0.6891$ & $9.9489$\\
\hline
$(m,a,b,\lambda)$ in\cite{YTCC} & $(20,5,5,2)$ & $15504$ & $15504$ & $0.7065$ & $50.0000$\\
$(m,\omega,2)$ in Theorem \ref{corollary-q=2-all} & $(14,7,2)$ & $16384$ & $9908$ & $0.6236$ & $44.3375$\\
\hline
$(m,a,b,\lambda)$ in\cite{YTCC} & $(20,8,7,3)$ & $125970$ & $77520$ & $0.6424$ & $273.0000$\\
$(m,\omega,2)$ in Theorem \ref{corollary-q=2-all} & $(17,8,2)$ & $131072$ & $65536$ &  $0.6291$ & $189.9219$ \\
\hline
\hline
$(m,a,b,\lambda)$ in\cite{YTCC} & $(16,4,5,2)$ & $1820$ & $4368$ & $0.6978$ & $10.0000$\\
$(m,\omega,3)$ in Theorem \ref{corollary-q=3-all} & $(7,5,3)$ & $2187$ & $1611$ & $0.5829$ & $28.5084$\\
\hline
$(m,a,b,\lambda)$ in\cite{YTCC} & $(16,5,6,1)$ & $4368$ & $8008$ & $0.7115$ & $10.0000$\\
$(m,\omega,3)$ in Theorem \ref{corollary-q=3-all} & $(8,6,3)$ & $6561$ & $5281$ & $0.6607$ & $34.7866$\\
\hline
$(m,a,b,\lambda)$ in\cite{YTCC} & $(17,10,7,5)$ & $19448$ & $19448$ & $0.7279$ & $21.0000$\\
$(m,\omega,3)$ in Theorem \ref{corollary-q=3-all} & $(9,7,3)$ & $19683$ & $16867$ & $0.7268$ & $42.0103$\\
\hline
$(m,a,b,\lambda)$ in\cite{YTCC} & $(20,5,6,3)$ & $15504$ & $38760$ & $0.8826$ & $4.0000$\\
$(m,\omega,3)$ in Theorem \ref{corollary-q=3-all} & $(9,8,3)$ & $19683$ & $19171$ & $0.8798$ & $9.2404$\\
\hline
\end{tabular}
\end{table}

\section{Conclusion}
\label{conclusion}
In this paper, we proposed a framework of constructing PDAs via Hamming distance. Consequently the problem of constructing PDAs is equivalent to  constructing appropriate partitions. According to the structure of obtained PDAs, we obtained two classes of coded caching schemes with linear subpacketizations. Finally theoretic and numerical comparisons showed that our new schemes have good performance.

In this paper, we pointed out that under the framework of Hamming distance, constructing coded caching scheme with small memory fraction, low subpacketization and small transmission rate depends on designing a partition such that 1) the cardinality of this partition is as small as possible and 2) the number of useless stars in each column of obtained PDA is numerable. So it is interesting to design the partitions satisfying the above two conditions.

\section*{Appendix A: The proof of Lemma \ref{lem-upper'}}
First the following notations are useful. A graph $G$ consists of a set $V(G)$ of vertexes and a set $E(G)\subset\{(u,v):u,v\in V(G)\}$ of edges.
The degree of a vertex $v$ in a graph $G$ is the number of vertices in $G$ that are adjacent to $v$.
The largest degree among the vertices of $G$ is called the maximum degree of $G$ is denoted by $\Delta(G)$.
A vertex $k$-coloring of a graph $G$ is an assignment of $k$ colors to the vertices of $G$, one color to each vertex, so that adjacent vertices are colored differently.
A graph $G$ is $k$-colorable if there exists a coloring of $G$ from a set of $k$ colors.
The minimum positive integer $k$ for which $G$ is $k$-colorable is the chromatic number of $G$ and is denoted by $\chi(G)$.

\begin{lemma}\label{degree}
\cite{CGT} For every graph $G$, $\chi(G)\leq1+\Delta(G)$.
\end{lemma}

Now let us give the proof of Lemma \ref{lem-upper'}.
\begin{proof}
For any $\mathcal{A}$ and $\mathcal{B}\subseteq[0,2)^m$, from Theorem \ref{th-main}, the PDA obtained by Construction \ref{construction}, Partition \ref{partition-q=2'} and Construction \ref{construction-2} has $S=S'l_{\mathbf{e}}$, where $S'$ is the number of vectors in array obtained from Construction \ref{construction}. Given $\mathcal{A}$ and $\mathcal{B}$ we only need to analyze the value of $l_{\mathbf{e}}$.
\begin{itemize}
  \item When $m=2\omega+1$, i.e., $m-\omega=\omega+1$, for any vector $\mathbf{e}$ occurring in $\mathbf{P}$ obtained by Construction \ref{construction}, and for any two vectors $\mathbf{t}_{j}$, $\mathbf{t}_{k}\in [0,2)^{m-\omega}$, from Partition \ref{partition-q=2'} we can see that $\mathcal{X}_{\mathbf{e},\mathbf{t}_{j}}$ can merge with $\mathcal{X}_{\mathbf{e},\mathbf{t}_{k}}$ if and only if $d(\mathbf{t}_{j},\mathbf{t}_{k})=\omega+1=m-\omega$, i.e., $\mathbf{t}_{j}+\mathbf{t}_{k}=\mathbf{1}$. Then every two vectors is a element by Partition \ref{partition-q=2'}. Hence we have $l_{\mathbf{e}}=2^{m-\omega-1}$.
  \item When $m>2\omega+1$, we turn the partition problem for $[0,2)^{m-\omega}$ to a vertex coloring problem. Define a graph $G$ with vertex set $V(G)=[0,2)^{m-\omega}$ such that there exists an edge connecting any two different vertices $\mathbf{t}_{1}$ and $\mathbf{t}_{2}$ in $V(G)$ if and only if $d(\mathbf{t}_{1}, \mathbf{t}_{2})\leq\omega$. For any vertex $\mathbf{\mathbf{t}}\in V(G)$, the number of vertices in $G$ that are adjacent to $\mathbf{t}$ is $\sum_{i=1}^{\omega}{m-\omega\choose i}$, i.e., the degree $\rho(\mathbf{t})\leq\sum_{i=1}^{\omega}{m-\omega\choose i}$. Then the maximum degree $\Delta(G)\leq \sum_{i=1}^{\omega}{m-\omega\choose i}$.
From Lemma \ref{degree}, we have
\begin{equation}\label{x}
  \chi(G)\leq1+\Delta(G)\leq1+\sum_{i=1}^{\omega}{m-\omega\choose i}.
\end{equation}

From the definition of $\chi(G)$, there exists a $\chi(G)$-coloring of $G$.
In fact the vertex $\chi(G)$-coloring of $G$ corresponds to a $\chi(G)$-partition for $[0,2)^{m-\omega}$ in Partition \ref{partition-q=2'}. In the vertex $\chi(G)$-coloring of graph $G$, we make each collection of vertexes having same color as a subset of vertex set $[0,2)^{m-\omega}$. Then there exists $\chi(G)$ subsets $\mathcal{D}_{0}$, $\mathcal{D}_{1}$, \ldots, $\mathcal{D}_{\chi(G)-1}$.
Since each vertex in $G$ has exactly one color, each element in $[0,2)^{m-\omega}$ is exactly contained in one subset.
For any two vertexes $\mathbf{t}_{1}$, $\mathbf{t}_{2} \in \mathcal{T}_{i}$, $i\in [0,\chi(G))$, there exist no edge $(\mathbf{t}_{1},\mathbf{t}_{2})$, i.e., $d(\mathbf{t}_{1}, \mathbf{t}_{2})\geq\omega+1$ which satisfying Partition \ref{partition-q=2'}. Hence $\{\mathcal{D}_{0}, \mathcal{D}_{1}, \ldots, \mathcal{D}_{\chi(G)-1}\}$ is a $\chi(G)$-partition for $[0,2)^{m-\omega}$.

From the above discussion, for the vertex $\chi(G)$-coloring of $G$, we always have a $l_{\mathbf{e}}=\chi(G)$-partition for $\mathcal{T}$. From \eqref{x} we have $l_{\mathbf{e}}\leq1+\sum_{i=1}^{\omega}{m-\omega\choose i}=\sum^{\omega}_{i=0}{m-\omega\choose i}$.
\end{itemize}
 Specially, when $\mathcal{A}=\mathcal{B}=[0,2)^m$, from the proof of Theorem \ref{corollary-q=2-all} we have $S'={m\choose \omega}$. Due to $S=S'l_{\mathbf{e}}$ the the proof of Lemma \ref{lem-upper'} is completed.
\end{proof}

\section*{Appendix B: The proof of Lemma \ref{upper}}
\label{proof-upper}
\begin{proof}
Similar to the proof of Lemma \ref{lem-upper'}, given $\mathcal{A}$, $\mathcal{B}\subseteq[0,3)^m$, from Theorem \ref{th-main}, the PDA obtained by Construction \ref{construction}, Partition \ref{partition-q=3'} and Construction \ref{construction-2} has $S=S'l_{\mathbf{e}}$, where $S'$ is the number of vectors in array obtained from Construction \ref{construction}. We tend to analyze the value of $l_{\mathbf{e}}$
and turn it to a vertex coloring problem.

Given set ${[0,m)\choose m-\omega}$ in Partition \ref{partition-q=3'},
define a graph $G$ with vertex set $V(G)={[0,m)\choose m-\omega}$ such that there exists an edge connecting any two different vertices $\mathcal{T}_{1}$ and
$\mathcal{T}_2$ in $V(G)$ if and only if $|\mathcal{T}_{1}\cap \mathcal{T}_{2}|\geq m-\frac{3\omega}{2}$.
For any vertex $\mathcal{T}\in V(G)$, the number of vertices in $G$ that are adjacent to $\mathcal{T}$ is $\sum^{m-\omega-1}_{ i=\lceil m-\frac{3\omega}{2}\rceil}{m-\omega\choose i}{\omega\choose m-\omega-i}$, i.e., the degree $\rho(\mathcal{T})\leq\sum^{m-\omega-1}_{ i=\lceil m-\frac{3\omega}{2}\rceil}{m-\omega\choose i}{\omega\choose m-\omega-i}$.
Then the maximum degree $\Delta(G)\leq \sum^{m-\omega-1}_{ i=\lceil m-\frac{3\omega}{2}\rceil}{m-\omega\choose i}{\omega\choose m-\omega-i}$.
From Lemma \ref{degree}, we have
\begin{equation}\label{x'}
\chi(G)\leq1+\Delta(G)\leq1+\sum^{m-\omega-1}_{ i=\lceil m-\frac{3\omega}{2}\rceil}{m-\omega\choose i}{\omega\choose m-\omega-i}.
\end{equation}

From the definition of $\chi(G)$, there exists a $\chi(G)$-coloring of $G$.
Similar to the proof of Lemma \ref{lem-upper'} in Appendix A, the vertex $\chi(G)$-coloring of $G$ corresponds to a $\chi(G)$-partition for ${[0,m)\choose m-\omega}$.

In the vertex $\chi(G)$-coloring of graph $G$, we make each collection of vertexes having same color as a subset of vertex set. Then there exists $\chi(G)$ subsets $\mathcal{D}_{0}$, $\mathcal{D}_{1}$, \ldots, $\mathcal{D}_{\chi(G)-1}$ of ${[0,m)\choose m-\omega}$.
Since each vertex in $G$ has exactly one color, each element in ${[0,m)\choose m-\omega}$ is exactly contained in one subset.
For any two vertexes $\mathcal{T}_{1}$, $\mathcal{T}_{2} \in \mathcal{D}_{i}$, $i\in [0,\chi(G))$, there exist no edge $(\mathcal{T}_{1},\mathcal{T}_{2})$, i.e., $|\mathcal{T}_{1}\cap\mathcal{T}_{2}|<m-\frac{3\omega}{2}$ which satisfying Partition \ref{partition-q=2'}. Hence $\{\mathcal{D}_{0}, \mathcal{D}_{1}, \ldots, \mathcal{D}_{\chi(G)-1}\}$ is a $\chi(G)$-partition for ${[0,m)\choose m-\omega}$.
For the vertex $\chi(G)$-coloring of $G$, we always have a $l_{\mathbf{e}}=\chi(G)$-partition for ${[0,m)\choose m-\omega}$. From \eqref{x'} we have $l_{\mathbf{e}}\leq1+\sum^{m-\omega-1}_{ i=\lceil m-\frac{3\omega}{2}\rceil}{m-\omega\choose i}{\omega\choose m-\omega-i}$.
When $\mathcal{A}=\mathcal{B}=[0,3)^m$, from the proof of Theorem \ref{corollary-q=3-all} we have $S'=3^m$. Due to $S=S'l_{\mathbf{e}}$ the proof is completed.
\end{proof}

\end{document}